\documentclass{compositio}
\usepackage[active]{srcltx}
\usepackage{amsmath}\usepackage{amssymb}\usepackage{amscd}\usepackage{amsthm}\usepackage{amsfonts}
\usepackage{cite}

\newcommand{\C}{\mathbb{C}}
\newcommand{\Z}{\mathbb{Z}}
\newcommand{\End}{\textrm{End}}
\newcommand{\Pe}{\textrm{P}}
\newcommand{\un}[1]{\underline{#1}}
\newcommand{\unn}[1]{\underline{\underline{#1}}}
\newcommand{\bc}{\textbf{c}}
\newcommand{\hR}{\widehat{R}}
\newcommand{\hF}{\widehat{F}}
\newcommand{\cT}{\mathcal{T}}
\newcommand{\sh}{\mathrm{sh}}
\newcommand{\cc}{\mathrm{cc}}
\newcommand{\qc}{\mathrm{c}^{(q)}}
\newcommand{\al}[1]{[#1]_{_{\alpha}}}
\newcommand{\ov}[1]{\overline{#1}}
\newcommand{\tP}{\widetilde{\textrm{P}}}

\newtheorem{theorem}{\textbf{Theorem}}[section]
\newtheorem{corollary}[theorem]{\textbf{Corollary}}
\newtheorem{lemma}[theorem]{\textbf{Lemma}}
\newtheorem{proposition}[theorem]{\textbf{Proposition}}
\newtheorem*{theo-intro}{\textbf{Theorem}}

\theoremstyle{definition}
\newtheorem{definition}[theorem]{\textbf{Definition}}
\newtheorem{example}[theorem]{\textbf{Example}}

\newtheorem{remark}[theorem]{\textbf{Remark}}

\numberwithin{equation}{section}

\begin{document} 

\title[Fusion formulas and fusion procedure]{Fusion formulas and fusion procedure for the Yang-Baxter equation}

\author{L. Poulain d'Andecy}
\email{loic.poulain-dandecy@univ-reims.fr}
\address{Universit\'e de Reims Champagne-Ardenne, UFR Sciences exactes et naturelles, Laboratoire de Math\'ematiques EA 4535 Moulin de la Housse BP 1039, 51100 Reims, France}



\thanks{}

\begin{abstract} 
We use the fusion formulas of the symmetric group and of the Hecke algebra to construct solutions of the Yang--Baxter equation on irreducible representations of $\mathfrak{gl}_N$, $\mathfrak{gl}_{N|M}$, $U_q(\mathfrak{gl}_N)$ and $U_q(\mathfrak{gl}_{N|M})$. The solutions are obtained via the fusion procedure for the Yang--Baxter equation, which is reviewed in a general setting. Distinguished invariant subspaces on which the fused solutions act are also studied in the general setting, and expressed, in general, with the help of a fusion function. Only then, the general construction is specialised to the four situations mentioned above. In each of these four cases, we show how the distinguished invariant subspaces are identified as irreducible representations, using the relevant fusion formula combined with the relevant Schur--Weyl duality.
\end{abstract}

\maketitle


\section{Introduction}\label{sec-intro}

\paragraph{\textbf{1.}} The Yang--Baxter equation originally appeared in statistical physics and in quantum integrable systems, and its study has then led to the discovery of quantum groups, which have found numerous applications in mathematics and mathematical physics (for literature on the Yang--Baxter equation and related subjects, see, \emph{e.g.}, \cite{CP,GRS,Is,Ji} and references therein).

The fusion procedure for the Yang--Baxter equation allows the construction of new solutions of the Yang--Baxter equation starting from a given fundamental solution. We refer to the new solutions as ``fused'' solutions of the Yang--Baxter equation. The fusion procedure has first been introduced in \cite{KRS}. Several aspects of the fusion procedure for the Yang--Baxter equation and of its applications to theoretical physics can be found for example in \cite{BaRe,Ch1,Ch3,DJKMO,DJMO,HZ,Ji2,Ma,PZ,Yu,Za}.

One main contribution of this paper is the application of the fusion procedure for the Yang--Baxter equation, in a setting corresponding to $\mathfrak{gl}_{N|M}$ and its quantum deformation $U_q(\mathfrak{gl}_{N|M})$. Already for $M=0$ (the more classical situation), we believe that the presentation made in this paper can be of interest. In addition, we emphasize that the first part of the paper is developed in a general setting, which might be then useful in a wider class of situations.

\paragraph{\textbf{2.}} In this paper we start with the description of the fusion procedure for the Yang--Baxter equation. The fusion procedure starts from a solution of the Yang--Baxter equation on a space $V$ and constructs a fused solution on tensor products of $V$ and their subspaces. We work here in a general setting; namely, we start from an arbitrary solution of the Yang--Baxter equation and consider arbitrary fusion parameters. In this general setting, we then focus on the identification and the study of distinguished invariant subspaces for the fused solutions. We provide complete proofs in this general setting. Besides, the construction is presented in a form which will allow us to apply fusion formulas (of the symmetric group and of the Hecke algebra). 

By construction, the aforementioned distinguished subspaces support a solution of the Yang--Baxter equation, by restriction of the fused solutions. Our next goal is to study in details these subspaces in particular cases where a fusion formula can be used. This is done for the classical Yang solution, its standard deformation and their ``super" analogues. This results in the identification of invariant subspaces as irreducible representations of, respectively, $\mathfrak{gl}_N$, $\mathfrak{gl}_{N|M}$, $U_q(\mathfrak{gl}_N)$ and $U_q(\mathfrak{gl}_{N|M})$, and in turn, in a construction, via the fusion procedure, of solutions of the Yang--Baxter equation acting on these representations. For the $\mathfrak{gl}_N$ and $\mathfrak{gl}_{N|M}$ cases, we present the ungraded situation independently and before the $\Z/2\Z$-graded one. We chose to do so for clarity of the exposition. Besides, our presentation of the $\mathfrak{gl}_N$ situation is made so that the other cases can be then dealt with very similarly. For the deformations, we give a unified presentation treating the general $U_q(\mathfrak{gl}_{N|M})$ situation, such that it is direct to obtain the particular $U_q(\mathfrak{gl}_N)$ case by setting $M=0$.

\paragraph{\textbf{3.}} We will describe now in more details the content of the paper. Let $V$ be a finite-dimensional complex vector space. Starting from an arbitrary solution $R(u)$ of the Yang--Baxter equation on $V\otimes V$ (see Section \ref{sec-prel} for more precision), a fused solution is built given two sequences of complex numbers $\bc=(c_1,\dots,c_n)$ and $\un{\bc}=(c_{\un{1}},\dots,c_{\un{n'}})$ and is denoted by $R_{\bc,\un{\bc}}(u)$. The resulting solutions of the Yang--Baxter equation $R_{\bc,\un{\bc}}(u)$ is an endomorphism of the vector space $V^{\otimes n}\otimes V^{\otimes n'}$. 

Then, linear subspaces $W_{\bc}\subset V^{\otimes n}$ and $W_{\un{\bc}}\subset V^{\otimes n'}$ are constructed, with the property that the fused solution $R_{\bc,\un{\bc}}(u)$ preserves the subspace $W_{\bc}\otimes W_{\un{\bc}}$ of $V^{\otimes n}\otimes V^{\otimes n'}\!$.  Thus, by restriction, the operator $R_{\bc,\un{\bc}}(u)$ induces a solution of the Yang--Baxter equation acting on the space $W_{\bc}\otimes W_{\un{\bc}}$. A crucial feature of the invariant subspaces is that, for any $\bc=(c_1,\dots,c_n)$, the subspace $W_{\bc}$ is defined as the image of an endomorphism $F(\bc)$ of $V^{\otimes n}$, and moreover the endomorphism $F(\bc)$ is expressed in terms of the original solution $R(u)$. We call the endomorphism $F(\bc)$ the fusion function.

The second part of the paper is devoted to the study of the endomorphisms $F(\bc)$, and in turn of the subspaces $W_{\bc}$, for some fundamental solutions $R(u)$. Our interest is to find some values of the parameters $\bc$ such that the evaluation of the fusion function $F(\bc)$ is not invertible and then to describe the image  $W_{\bc}$. This is where the so-called ``fusion formulas'' will play a major role. 

For the symmetric group $S_n$, the fusion formula expresses each element of a complete set of primitive idempotents of $\C S_n$ as a certain evaluation of a rational function in several variables with coefficients in $\C S_n$ (the rational function is the same for all the idempotents). A generalization of the fusion formula of $S_n$ exists also for the Hecke algebra, and we will use both of them. As the name suggests, these fusion formulas are intimately related to the fusion procedure. Indeed the rational function involved in these fusion formulas turn out to be closely related to the fusion functions $F(\bc)$ mentioned above, in each of the four cases considered in this paper. For each case, it is our main goal in the second part of the paper to explain how the statements of the fusion formulas can be translated into properties of the fusion function, resulting in the identification of interesting distinguished subspaces $W_{\bc}$.

We start with the case where $R(u)$ is the so-called Yang solution. We give all details in this case to use it as a reference in the sequel. Then we generalize the obtained results to the case where $R(u)$ is a generalization of the Yang solution for a $\Z/2\Z$-graded vector space $V$ (\emph{i.e.} a ``super" vector space), and to the case where $R(u)$ is the standard deformation of the Yang solution, both in the usual and $\Z/2\Z$-graded situations.

{}For the Yang solution, the main result  is expressed as the identification of subspaces $W_{\bc}$ which are isomorphic to finite-dimensional irreducible representations of the general linear Lie algebra $\mathfrak{gl}_N$, where $N$ is the dimension of $V$. It results in a family of fused solutions of the Yang--Baxter equation acting on finite-dimensional irreducible representations of $\mathfrak{gl}_N$. Let $\lambda$ be a partition of length less than or equal to $N$. A subspace $W_{\bc}$ isomorphic to the irreducible representation of highest weight $\lambda$ is obtained when the sequence of numbers $\bc$ is a sequence of classical contents associated to a standard Young tableau of shape $\lambda$. We show moreover that, by replacing the set $\bc$ by a sequence of contents associated to another standard Young tableau of the same shape, we obtain an equivalent fused solution up to a change of basis.

As it is suggested in \cite{Mo}, the study of the subspaces $W_{\bc}$ when $R(u)$ is the Yang solution is made using the classical Schur--Weyl duality (see, \emph{e.g.}, \cite{GW,We}) together with the so-called ``fusion formula'' for the symmetric group. The fusion formula for the symmetric group $S_n$ originates in \cite{Ju} and has then been developed in \cite{Ch2}, see also \cite{Mo,Na,Na2,Na3}.

When $R(u)$ is the generalization of the Yang solution for a $\Z/2\Z$-graded vector space $V=V_{\ov{0}}\oplus V_{\ov{1}}$, we still can use the fusion formula for the symmetric group $S_n$, even if, in this situation, another action of $S_n$ on $V^{\otimes n}$ is relevant. This action gives rise to a ``super" analogue of the classical Schur--Weyl duality \cite{BeRe,Se}, which holds between the symmetric group and the general linear Lie superalgebra $\mathfrak{gl}_{N|M}$, where $N:=\dim(V_{\ov{0}})$ and $M:=\dim(V_{\ov{1}})$. However, we are still able to relate the fusion function $F(\bc)$ in this situation with the rational function used in the fusion formula of the symmetric group.
This allows us to identify the invariant subspaces and to obtain a family of fused solutions of the Yang--Baxter equation acting on finite-dimensional irreducible representations of $\mathfrak{gl}_{N|M}$.

{}For the standard deformation of the Yang solution and its generalization for a $\Z/2\Z$-graded vector space, we use the fusion formula for the Hecke algebra \cite{IMOs} together with the quantum analogue of the Schur--Weyl duality. This duality holds between the Hecke algebra and the quantum group $U_q(\mathfrak{gl}_{N|M})$ associated to the Lie superalgebra $\mathfrak{gl}_{N|M}$ (see \cite{Ji1} for the ungraded situation $M=0$, and \cite{Mi,Moon} for the generalisation to $M>0$). Here again, we can relate the fusion function $F(\bc)$ with the rational function used in the fusion formula for the Hecke algebra, and identify the invariant subspaces. Thus, via the fusion procedure, we obtain a family of fused solutions of the Yang--Baxter equation acting on finite-dimensional irreducible representations of $U_q(\mathfrak{gl}_{N|M})$. The parameters $\bc$ are now related to quantum contents of standard Young tableaux. Here again, replacing a standard Young tableau by another standard Young tableau of the same shape leads to equivalent fused solutions up to a change of basis.

\paragraph{\textbf{4.}} In the first part of the paper, for an arbitrary fundamental solution $R(u)$, we study the possible equivalence of the fused solutions under a permutation of the parameters $\bc=(c_1,\dots,c_n)$. By ``equivalence'' of two fused solutions, we mean that the restrictions of the fused solutions to the distinguished subspaces are related by a change of basis. Assuming an unitary property for $R(u)$, we obtain that particular permutations (that we call ``admissible'' for $\bc$ and $R(u)$, see Subsection \ref{subsec-dep}) of the parameters $\bc=(c_1,\dots,c_n)$ lead to equivalent fused solutions. For the four solutions $R(u)$ considered above, the fact that the fused solutions depend on the standard Young tableaux only through their shapes is a consequence of this general property.

We note that, even for solutions $R(u)$ considered in this paper, the study of (proper) invariant subspaces for the fused solutions is, though quite satisfactory, not fully complete. Namely, it does not seem to be known how to determine all the sets of parameters $\bc$ such that the fusion function $F(\bc)$ is non-invertible (and thus leads to a proper invariant subspace) and to study the obtained fused solutions. Even the classification of the sets of parameters $\bc$ such that the operator $F(\bc)$ is proportional to a projector does not seem to be completely understood.

We note that analogues of fusion formulas of the symmetric group and of the Hecke algebra have been obtained for several other structures: the Brauer algebras \cite{IM,IMO1}, the Birman--Wenzl--Murakami algebras \cite{IMO2}, the complex reflection groups of type $G(d,1,n)$ \cite{OPdA1}, the Ariki--Koike algebras \cite{OPdA2} and the wreath products of finite groups by the symmetric group \cite{PdA}. These fusion formulas involve as well a rational function (with coefficients in the algebras under consideration) which is built from elementary ``universal" solutions of the Yang--Baxter equation, the analogues for these structures of the Baxterized elements of $\C S_n$. However, in these fusion formulas, Baxterized solutions of the reflection equation also appear in the rational functions. It is expected, as it is already mentioned in \cite{IMO2}, that these fusion formulas admit an interpretation in the framework of the fusion procedure for the reflection equation.

\paragraph{\textbf{5.}} The paper is organized as follows. In Section \ref{sec-prel}, we give necessary definitions and notations. 

In Section \ref{sec-fus}, starting from an arbitrary solution $R(u)$ of the Yang--Baxter equation acting on $V\otimes V$, we describe the construction of the endomorphisms $R_{\bc,\un{\bc}}(u)$ and prove that the obtained operators $R_{\bc,\un{\bc}}(u)$ are solutions of the Yang--Baxter equation as well. 

In Section \ref{sec-inv}, we define the fusion function $F(\bc)$ and show that its image provides invariant subspaces for the fused solutions. We also give an alternative formula for the fusion function $F(\bc)$ and prove the equivalence of the fused solutions under admissible permutation of the parameters $\bc$. Both Sections \ref{sec-fus} and \ref{sec-inv} deal with an arbitrary solution $R(u)$ of the Yang--Baxter equation (only the unitary property is assumed in Subsection \ref{subsec-dep}). 

In Section \ref{sec-Sn}, we apply the whole procedure to the situation where $R(u)$ is the Yang solution, and identify invariant subspaces isomorphic to finite-dimensional irreducible representations of the general linear Lie algebra. The main results are summarized in Corollary \ref{coro-fin-Sn}. We also provide explicit examples of fused solutions and describe some ``non-standard" evaluations of the fusion function.

In Sections \ref{sec-Sn-s} and \ref{sec-Hn}, we generalize the results of the preceding section to the situation where $R(u)$ is a generalization of the Yang solution for a super vector space and where $R(u)$ is the standard deformation of the Yang solution and of its ``super" generalization.

\section{Preliminaries}\label{sec-prel}

\paragraph{\textbf{1.}} Let $a$ and $b$ be two integers such that $a\leq b$. We will use the following notation for the product of non-commuting quantity $x_i$ depending on $i$: 
$$\displaystyle\prod_{i=a,\dots,b}^{\rightarrow}x_i:=x_ax_{a+1}\dots x_b\ \quad\text{and}\ \quad\displaystyle\prod_{i=a,\dots,b}^{\leftarrow}x_i:=x_b\dots x_{a+1}x_a\ .$$
Similarly, for non-commuting quantities depending on two indices $i$ and $j$, the arrow ``$\rightarrow$" indicates that the factors are ordered lexicographically, while the arrow ``$\leftarrow$" indicates that the factors are ordered in the reverse lexicographic order. For example, we have:
$$\displaystyle\prod_{1\leq i<j\leq 4}^{\rightarrow}x_{ij}:=x_{12}x_{13}x_{14}x_{23}x_{24}x_{34}\ \quad\text{and}\ \quad\displaystyle\prod_{1\leq i<j\leq 4}^{\leftarrow}x_{ij}:=x_{34}x_{24}x_{23}x_{14}x_{13}x_{12}\ .$$

\paragraph{\textbf{2.}} The product of two (and any number of) permutations is written with the standard notation for the composition of functions, namely, if $\pi,\sigma\in S_n$ are two permutations then $\pi\sigma$ means that $\sigma$ is applied first and then $\pi$; for example, $(1,2)(2,3)=(1,2,3)$ in the standard cyclic notation.

\paragraph{\textbf{3.}} For a vector space $V$ over $\C$, we denote by $\End(V)$ the set of endomorphisms of $V$. We fix a finite-dimensional vector space $V$ over $\C$ and let $\textrm{Id}$ denote the identity endomorphism of $V$ and $\Pe$ denote the permutation endomorphism of $V\otimes V$ (that is, $\Pe(x\otimes y):=y\otimes x$ for $x,y\in V$). We will denote by $\textrm{Id}_{V^{\otimes k}}$ the identity operator on $V^{\otimes k}$, for $k\in\Z_{\geq2}$.

We will use the standard notation for operators in $\End(V^{\otimes n})$; namely, if $T\in\End(V)$ then $T_i$, for $i=1,\dots,n$, will denote the operator in $\End(V^{\otimes n})$ acting as $T$ on the $i$-th copy and trivially anywhere else. For example, for $n=3$, we have
$$T_1:=T\otimes\textrm{Id}_{V^{\otimes 2}}\,,\ \quad\ T_2:=\textrm{Id}\otimes T\otimes\textrm{Id}\ \quad\text{and}\ \quad T_3:=\textrm{Id}_{V^{\otimes 2}}\otimes T\ .$$
Similarly, if $R$ is an operator in $\End(V\otimes V)$ then $R_{i,j}$ will denote the operator in $\End(V^{\otimes n})$ acting as $R$ in the $i$-th and $j$-th copies and trivially anywhere else. For example, if $n=3$:
$$R_{1,2}:=R\otimes\textrm{Id}\,,\ \quad\ R_{2,3}:=\textrm{Id}\otimes R\ \quad\text{and}\ \quad R_{1,3}:=(\textrm{Id}\otimes \Pe)R_{1,2}(\textrm{Id}\otimes \Pe)=\Pe_{2,3}R_{1,2}\Pe_{2,3}\ .$$ 
We also have for example $R_{2,1}:=\Pe_{1,2}R_{1,2}\Pe_{1,2}$. Note that we slightly abuse by not indicating the integer $n$ in the notation $T_i$, $R_{i,j}$, $\Pe_{i,j}$, etc. This should not lead to any confusion as it will be clear on which space the operators act.

We recall some obvious commutation relations that we will often use throughout the text without mentioning:
$$\begin{array}{rcll}
T_iT_j& = &T_jT_i 	& \text{if $i\neq j$\,,}\\[0.2em]
T_iR_{j,k}&=&R_{j,k}T_i & \text{if $i\notin\{j,k\}$\,,}\\[0.2em]
R_{i,j}R_{k,l}&=&R_{k,l}R_{i,j}& \text{if $\{i,j\}\cap\{k,l\}=\emptyset$\,,}\\[0.2em]
\Pe_{i,j}T_k&=&T_{s_{i,j}(k)}\Pe_{i,j} & \text{for all $i,j,k=1,\dots,n$ with $i\neq j$,}\\[0.2em]
\Pe_{i,j}R_{k,l}&=&R_{s_{i,j}(k),s_{i,j}(l)}\Pe_{i,j}\ \ \  & \text{for all $i,j,k,l=1,\dots,n$ with $i\neq j$ and $k\neq l$,}
\end{array} $$
where $s_{i,j}$ denotes the transposition of $i$ and $j$.

\paragraph{\textbf{4. Yang--Baxter equation.}} Let $R$ be a function of one variable $u\in\C$, taking values in $\End(V\otimes V)$. The variable $u$ is called the \emph{spectral parameter}. The function $R$ is a solution, on $V\otimes V$, of the Yang--Baxter equation if the following functional relation is satisfied:
\begin{equation}\label{YB-add1}
R_{1,2}(u)R_{1,3}(u+v)R_{2,3}(v)=R_{2,3}(v)R_{1,3}(u+v)R_{1,2}(u)\ ,
\end{equation}
where both sides take values in $\End(V^{\otimes 3})$ and $R_{1,2}(u)$ means $R(u)_{1,2}$\,, etc. (we will always use this standard notation). By abuse of speaking, we will sometimes say that the operator $R(u)$ itself is a solution of the Yang--Baxter equation.

More generally, let $E$ be an indexing set, $\{V_{\mu}\}_{\mu\in E}$ be a collection of finite-dimensional vector spaces over $\C$, and let $\{R_{\mu,\nu}\}_{\mu,\nu\in E}$ be a family of functions of $u\in\C$ such that $R_{\mu,\nu}$ takes values in $\End(V_{\mu}\otimes V_{\nu})$, for any $\mu,\nu\in E$. We say that the functions $R_{\mu,\nu}$ form a family of solutions of the Yang--Baxter equation if the following functional relations are satisfied:
\begin{equation}\label{YB-add-fam}
R_{\mu,\nu}(u)R_{\mu,\tau}(u+v)R_{\nu,\tau}(v)=R_{\nu,\tau}(v)R_{\mu,\tau}(u+v)R_{\mu,\nu}(u)\ \ \ \ \ \text{for any $\mu,\nu,\tau\in E$,}
\end{equation}
where both sides take values in $\End(V_{\mu}\otimes V_{\nu}\otimes V_{\tau})$ (in (\ref{YB-add-fam}), $R_{\mu,\nu}(u)$ stands for $R_{\mu,\nu}(u)\otimes\textrm{Id}_{V_{\tau}}$, and similarly for $R_{\mu,\tau}(u+v)$ and $R_{\nu,\tau}(v)$).

Equation (\ref{YB-add1}) is the Yang--Baxter equation with ``additive" spectral parameters. We will sometimes use the version of the equation with ``multiplicative" spectral parameters, namely
\begin{equation}\label{YB-mult1}
R_{1,2}(\alpha)R_{1,3}(\alpha\beta)R_{2,3}(\beta)=R_{2,3}(\beta)R_{1,3}(\alpha\beta)R_{1,2}(\alpha)\ ,
\end{equation}
for a function $R$ of the variable $\alpha\in\C$ taking values in $\End(V\otimes V)$.

\paragraph{\textbf{5. Partitions and Young tableaux.}} Let $\lambda\vdash n$ be a partition of a positive integer $n$, that is, $\lambda=(\lambda_1,\dots,\lambda_l)$ is a family of  integers such that $\lambda_1\geq\lambda_2\geq\dots\geq\lambda_l>0$ and $\lambda_1+\dots+\lambda_l=n$. We say that $\lambda$ is a partition {\em of size} $n$ and {\em of length} $l$ and set $|\lambda|:=n$ and $\ell(\lambda):=l$. 

The Young diagram of $\lambda$ is the set of elements $(x,y)\in\mathbb{Z}^2$ such that $x\in\{1,\dots,l\}$ and $y\in\{1,\dots,\lambda_x\}$. A pair $(x,y)\in\mathbb{Z}^2$ is usually called a {\em node}. The Young diagram of $\lambda$ will be represented in the plan by a left-justified array of $l$ rows such that the $j$-th row contains $\lambda_j$ nodes for all $j=1,\dots,l$ (a node will be pictured by an empty box). We number the rows from top to bottom. We identify partitions with their Young diagrams and say that $(x,y)$ is a node of $\lambda$, or $(x,y)\in\lambda$, if $(x,y)$ is a node of the diagram of $\lambda$.

For a node $\theta=(x,y)$, we set $\cc(\theta):=y-x$. The number $\cc(\theta)$ is called the \emph{classical content} of $\theta$. For any complex number $q$, we set $\qc(\theta):=q^{2(y-x)}$ and call $\qc(\theta)$ the \emph{$q$-quantum content} of $\theta$, or simply the \emph{quantum content} of $\theta$ when $q$ is fixed.

The hook of a node $\theta\in\lambda$ is the set of nodes of $\lambda$ consisting of the node $\theta$ and the nodes which lie either under $\theta$ in the same column or to the right of $\theta$ in the same row; the hook length $h_{\lambda}(\theta)$ of $\theta$ is the cardinality of the hook of $\theta$. We define, for any partition $\lambda$,
\begin{equation}\label{f-lam}
f(\lambda):=\Bigl(\prod_{\theta\in\lambda}h_{\lambda}(\theta)\Bigr)^{-1}.
\end{equation}
For any non-zero complex number $q$, we also define for a partition $\lambda$:
\begin{equation}\label{fq-lam}
f^{(q)}(\lambda):=\prod_{\theta\in\lambda}\frac{q^{\cc(\theta)}}{[h_{\lambda}(\theta)]}\ ,
\end{equation}
where $\displaystyle [N]:=\frac{q^N-q^{-N}}{q-q^{-1}}$ for $N\in\mathbb{Z}\,$.

A Young tableau $\mathcal{T}$ of shape $\lambda$ is a bijection between the set $\{1,\dots,n\}$ and the set of nodes of $\lambda$. In other words, a Young tableau of shape $\lambda$ is obtained by placing without repetition the numbers $1,\dots,n$ into the nodes of $\lambda$. We use the notation $\sh_{\cT}$ to denote the shape of $\cT$. A Young tableau $\cT$ of shape $\lambda$ is {\em standard} if its entries  increase along any row and down any column of  $\lambda$, that is, the number corresponding to the node $(x,y)$ is greater than the number corresponding to the node $(x',y')$ when $x\geq x'$ and $y\geq y'$.

For a Young tableau ${\mathcal{T}}$, we denote respectively by $\cc({\mathcal{T}}|i)$ and $\qc({\mathcal{T}}|i)$ the classical content and the $q$-quantum content of the node containing the number $i$, for all $i=1,\dots,|\sh_{\cT}|$.

\section{Fused solutions of the Yang--Baxter equation}\label{sec-fus}

We fix $V$ a finite-dimensional vector space over $\C$ and let $R$ be a function of $u\in\C$, taking values in $\End(V\otimes V)$, which satisfies the Yang--Baxter equation:
\begin{equation}\label{YB-add}
R_{1,2}(u)R_{1,3}(u+v)R_{2,3}(v)=R_{2,3}(v)R_{1,3}(u+v)R_{1,2}(u)\ ,
\end{equation}
where both sides operate on $V\otimes V\otimes V$ (in Sections \ref{sec-fus} and \ref{sec-inv}, we work with the additive convention for the spectral parameters, see paragraph \textbf{4} of Section \ref{sec-prel}; equivalently, we could have chosen the multiplicative version, as indicated in Remark \ref{rem-mult} below).

\begin{example}\label{ex-Yang}
One of the simplest examples of a solution of Equation (\ref{YB-add}) is the Yang solution:
\begin{equation}\label{Yang}R(u)=\textrm{Id}_{V^{\otimes 2}}-\frac{\Pe}{u}\ .\end{equation}
The fact that the function given by (\ref{Yang}) satisfies the Yang--Baxter equation (\ref{YB-add}) can be checked by a direct calculation.\hfill$\triangle$
\end{example}

Let $n$ and $n'$ be positive integers. We consider the space $V^{\otimes n}\otimes V^{\otimes n'}$ and label the copies of $V$ by $1,\dots,n,\un{1},\dots,\un{n'}$ (from left to right) such that, for example, the operator $R_{a,\un{b}}(u)$ with $a\in\{1,\dots,n\}$ and $b\in\{1,\dots,n'\}$ stands for the operator $R_{a,n+b}(u)$ with the notation explained in Section \ref{sec-prel} (we prefer the underlined notation to make a clear distinction between the indices corresponding to $V^{\otimes n}$ and the indices corresponding to $V^{\otimes n'}$). 

Let $\bc:=(c_1,\dots,c_n)$ be an $n$-tuple of complex parameters and $\un{\bc}:=(c_{\un{1}},\dots,c_{\un{n'}})$ be an $n'$-tuple of complex parameters. We define a function $R_{\bc,\un{\bc}}$ of one variable $u\in\C$ taking values in $\End(V^{\otimes n}\otimes V^{\otimes n'})$ by:
\begin{equation}\label{def-fus-R}
R_{\bc,\un{\bc}}(u):=\prod_{i=1,\dots,n'}^{\rightarrow}R_{n,\un{i}}(u+c_n-c_{\un{i}})\dots\dots R_{2,\un{i}}(u+c_2-c_{\un{i}})R_{1,\un{i}}(u+c_1-c_{\un{i}})\ .
\end{equation}
Moving all the operators with $n$ as a first index to the left and repeating the process for $n-1,n-2,\dots,2$, we find the following alternative form for $R_{\bc,\un{\bc}}(u)$:
\begin{equation}\label{def-fus-R2}
R_{\bc,\un{\bc}}(u)=\prod_{i=1,\dots,n}^{\leftarrow}R_{i,\un{1}}(u+c_i-c_{\un{1}})R_{i,\un{2}}(u+c_i-c_{\un{2}})\dots\dots R_{i,\un{n'}}(u+c_i-c_{\un{n'}})\ .
\end{equation}

We prove in the following theorem that the set of functions $\{R_{\bc,\un{\bc}}\}
$, where $\bc\in\C^n$, $\un{\bc}\in\C^{n'}$ and $n,n'>0$,
forms a family of solutions of the Yang--Baxter equation. We will call elements of this family ``fused solutions" of the Yang--Baxter equation, and operators $R_{\bc,\un{\bc}}(u)$ ``fused operators".
\begin{theorem}\label{thm-fus-R}
Let $n$, $n'$ and $n''$ be positive integers and let $\bc:=(c_1,\dots,c_n)$, $\un{\bc}:=(c_{\un{1}},\dots,c_{\un{n'}})$ and $\unn{\bc}:=(c_{\unn{1}},\dots,c_{\unn{n''}})$ be, respectively, an $n$-tuple, an $n'$-tuple and an $n''$-tuple of complex parameters.
We have the functional equation
\begin{equation}\label{fus-YB-add}
R_{\bc,\un{\bc}}(u)R_{\bc,\unn{\bc}}(u+v)R_{\un{\bc},\unn{\bc}}(v)=R_{\un{\bc},\unn{\bc}}(v)R_{\bc,\unn{\bc}}(u+v)R_{\bc,\un{\bc}}(u)\ ,
\end{equation}
where both sides take values in $\End(V^{\otimes n}\otimes V^{\otimes n'}\otimes V^{\otimes n''})$ and the copies of $V$ are labelled by $1,\dots,n$, $\un{1},\dots,\un{n'}$, $\unn{1},\dots,\unn{n''}$ (from left to right). In Equation (\ref{fus-YB-add}), $R_{\bc,\un{\bc}}(u)$ stands for the operator $R_{\bc,\un{\bc}}(u)\otimes\textrm{Id}_{V^{\otimes n''}}$ and similarly for $R_{\bc,\unn{\bc}}(u+v)$ and $R_{\un{\bc},\unn{\bc}}(v)$.
\end{theorem}
\begin{proof}
Note first that the equality $(u+c_i-c_{\un{j}})+(v+c_{\un{j}}-c_{\unn{k}})=u+v+c_i-c_{\unn{k}}$, valid for any $i\in\{1,\dots,n\}$, $j\in\{1,\dots,n'\}$ and $k\in\{1,\dots,n''\}$, ensures that, in the product of operators in the left hand side of (\ref{fus-YB-add}), we can always apply the Yang--Baxter equation (\ref{YB-add}) as soon as the indices are properly arranged. In other words, we will never have to worry about the spectral parameters. Therefore, for saving place during the proof, we will drop them out of the notation, namely we set, until the end of the proof, 
$$R_{i,\un{j}}:=R_{i,\un{j}}(u+c_i-c_{\un{j}})\,,\ \quad\ R_{i,\unn{k}}:=R_{i,\unn{k}}(u+v+c_i-c_{\unn{k}})\ \quad\text{and}\ \quad R_{\un{j},\unn{k}}:=R_{\un{j},\unn{k}}(v+c_{\un{j}}-c_{\unn{k}})\ ,$$
for any $i\in\{1,\dots,n\}$, $j\in\{1,\dots,n'\}$ and $k\in\{1,\dots,n''\}$.

We prove the formula (\ref{fus-YB-add}) by induction on $n$. If $n=1$ then the left hand side of (\ref{fus-YB-add}) is
$$\begin{array}{cl}
& R_{1,\un{1}}R_{1,\un{2}}\dots R_{1,\un{n'}}\cdot R_{1,\unn{1}}R_{1,\unn{2}}\dots R_{1,\unn{n''}}\cdot \displaystyle\prod_{i=1,\dots,n''}^{\rightarrow}\Bigl(R_{\un{n'},\unn{i}}\dots R_{\un{2},\unn{i}}R_{\un{1},\unn{i}}\Bigr)\\
= & R_{1,\un{1}}R_{1,\un{2}}\dots R_{1,\un{n'}}\cdot \displaystyle\prod_{i=1,\dots,n''}^{\rightarrow}\Bigl(R_{1,\unn{i}}\cdot R_{\un{n'},\unn{i}}\dots R_{\un{2},\unn{i}}R_{\un{1},\unn{i}}\Bigr)\ .
\end{array}$$
We have, for $i=1,\dots,n''$,
$$\begin{array}{cl}
& R_{1,\un{1}}R_{1,\un{2}}\dots R_{1,\un{n'}}\cdot R_{1,\unn{i}}\cdot R_{\un{n'},\unn{i}}\dots R_{\un{2},\unn{i}}R_{\un{1},\unn{i}}\\[0.4em]
= \!\!\!&  R_{1,\un{1}}R_{1,\un{2}}\dots R_{1,\un{n'-1}}\cdot R_{\un{n'},\unn{i}}R_{1,\unn{i}}R_{1,\un{n'}}\cdot R_{\un{n'-1},\unn{i}}\dots R_{\un{2},\unn{i}}R_{\un{1},\unn{i}}\,\quad\text{(as $R_{1,\un{n'}}R_{1,\unn{i}}R_{\un{n'},\unn{i}}=R_{\un{n'},\unn{i}}R_{1,\unn{i}}R_{1,\un{n'}}\,$)}\\[0.4em]
= \!\!\!&  R_{\un{n'},\unn{i}}\cdot R_{1,\un{1}}R_{1,\un{2}}\dots R_{1,\un{n'-1}}\cdot R_{1,\unn{i}}\cdot R_{\un{n'-1},\unn{i}}\dots R_{\un{2},\unn{i}}R_{\un{1},\unn{i}}\cdot R_{1,\un{n'}}\ .
\end{array}$$
Repeating this calculation with $n'$ replaced by $n'-1$ and so on, we find that
$$R_{1,\un{1}}R_{1,\un{2}}\dots R_{1,\un{n'}}\cdot R_{1,\unn{i}}\cdot R_{\un{n'},\unn{i}}\dots R_{\un{2},\unn{i}}R_{\un{1},\unn{i}}=R_{\un{n'},\unn{i}}\dots R_{\un{2},\unn{i}}R_{\un{1},\unn{i}}\cdot R_{1,\unn{i}}\cdot R_{1,\un{1}}R_{1,\un{2}}\dots R_{1,\un{n'}}\ .$$
So we finally obtain that the left hand side of (\ref{fus-YB-add}) for $n=1$ is equal to
$$\begin{array}{cl}
&\displaystyle\prod_{i=1,\dots,n''}^{\rightarrow}\Bigl(R_{\un{n'},\unn{i}}\dots R_{\un{2},\unn{i}}R_{\un{1},\unn{i}}\cdot R_{1,\unn{i}}\Bigr)\cdot R_{1,\un{1}}R_{1,\un{2}}\dots R_{1,\un{n'}}\\
= & \displaystyle\prod_{i=1,\dots,n''}^{\rightarrow}\Bigl(R_{\un{n'},\unn{i}}\dots R_{\un{2},\unn{i}}R_{\un{1},\unn{i}}\Bigr)\cdot R_{1,\unn{1}}R_{1,\unn{2}}\dots R_{1,\unn{n''}}\cdot R_{1,\un{1}}R_{1,\un{2}}\dots R_{1,\un{n'}}\ ,
\end{array}$$
which coincides with the right hand side of (\ref{fus-YB-add}) for $n=1$.

Now let $n>1$ and set $\bc^{(n-1)}:=(c_1,\dots,c_{n-1})$. Using (\ref{def-fus-R2}) and commutation relations, we reorganize the left hand side of (\ref{fus-YB-add}) and write it as
$$ R_{n,\un{1}}R_{n,\un{2}}\dots R_{n,\un{n'}}\cdot R_{n,\unn{1}}R_{n,\unn{2}}\dots R_{n,\unn{n''}}\cdot R_{\bc^{(n-1)}\!,\un{\bc}}(u)\,R_{\bc^{(n-1)}\!,\unn{\bc}}(u+v)\,R_{\un{\bc},\unn{\bc}}(v)\ .$$
We use the induction hypothesis to transform this expression into
$$ R_{n,\un{1}}R_{n,\un{2}}\dots R_{n,\un{n'}}\cdot R_{n,\unn{1}}R_{n,\unn{2}}\dots R_{n,\unn{n''}}\cdot R_{\un{\bc},\unn{\bc}}(v)\,R_{\bc^{(n-1)}\!,\unn{\bc}}(u+v)\,R_{\bc^{(n-1)}\!,\un{\bc}}(u)\ .$$
Then we use the induction basis (with the space labelled here by $n$ playing the same role as the space labelled by $1$ in the calculation for the induction basis) to move $R_{\un{\bc},\unn{\bc}}(v)$ to the left and we obtain
$$ R_{\un{\bc},\unn{\bc}}(v)\cdot R_{n,\unn{1}}R_{n,\unn{2}}\dots R_{n,\unn{n''}}\cdot R_{n,\un{1}}R_{n,\un{2}}\dots R_{n,\un{n'}} \cdot R_{\bc^{(n-1)}\!,\unn{\bc}}(u+v)\,R_{\bc^{(n-1)}\!,\un{\bc}}(u)\ .$$
This is equal, using again commutation relations, to $R_{\un{\bc},\unn{\bc}}(v)R_{\bc,\unn{\bc}}(u+v)R_{\bc,\un{\bc}}(u)$, that is, to the right hand side of (\ref{fus-YB-add}).
\end{proof}

\begin{remark}\label{rem-graph}
To the operator $R(u)$ is associated the following figure:
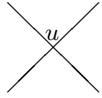
\begin{center}
\setlength{\unitlength}{0.004cm}
\begin{picture}(300,300)(0,-300)
\put(0,0){ \line(1, -1){300}}
\put(300,0){ \line(-1, -1){300}}
\put(155,-125){\footnotesize{$u$}}
\end{picture}
\end{center}
Roughly speaking, this represents the interaction of two particles at the intersection of the two lines, and the interaction is governed, in some sense, by the operator $R(u)$ (see, \emph{e.g.}, \cite{GRS} for the physical meaning of the graphical formulation of the Yang--Baxter equation). Then, the Yang--Baxter equation (\ref{YB-add}) is formulated graphically as
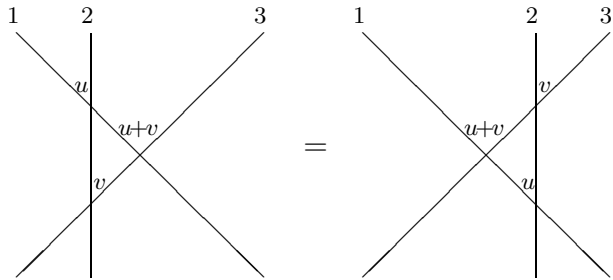
\begin{center}
\setlength{\unitlength}{0.0065cm}
\begin{picture}(1200,550)(0,-500)
\put(0,20){\footnotesize{1}}
\put(150,20){\footnotesize{2}}
\put(500,20){\footnotesize{3}}
\put(0,0){ \line(1, -1){500}}
\put(150,0){ \line(0, -1){500}}
\put(500,0){ \line(-1, -1){500}}
\put(600,-250){\large{$=$}}
\put(137,-120){\footnotesize{$u$}}
\put(225,-210){\footnotesize{$u\!\!+\!\!v$}}
\put(175,-320){\footnotesize{$v$}}
\put(700,20){\footnotesize{1}}
\put(1050,20){\footnotesize{2}}
\put(1200,20){\footnotesize{3}}
\put(700,0){ \line(1, -1){500}}
\put(1050,0){ \line(0, -1){500}}
\put(1200,0){ \line(-1, -1){500}}
\put(1075,-120){\footnotesize{$v$}}
\put(925,-210){\footnotesize{$u\!\!+\!\!v$}}
\put(1040,-320){\footnotesize{$u$}}
\end{picture}
\end{center}
The numbers above the lines are the indices of the copy of $V$ in $V^{\otimes 3}$. By convention, we read the figure from top to bottom and we write the corresponding operators from left to right; for example, the left hand side of the picture above corresponds to $R_{1,2}(u)R_{1,3}(u+v)R_{2,3}(v)$.

Within this graphical interpretation, the fused solutions, given by Formula (\ref{def-fus-R}), correspond to interactions of multiplets of particles, namely, $n$ particles interacting with $n'$ particles. As an example, for $n=n'=2$, the corresponding interaction is depicted by:
\begin{center}
\setlength{\unitlength}{0.008cm}
\begin{picture}(650,530)(0,-500)
\put(0,20){\footnotesize{1}}
\put(150,20){\footnotesize{2}}
\put(500,20){\footnotesize{$\un{1}$}}
\put(650,20){\footnotesize{$\un{2}$}}
\put(0,0){ \line(1, -1){500}}
\put(150,0){ \line(1, -1){500}}
\put(500,0){ \line(-1, -1){500}}
\put(650,0){ \line(-1, -1){500}}
\put(230,-215){\footnotesize{$u_{1\un{1}}$}}
\put(380,-215){\footnotesize{$u_{2\un{2}}$}}
\put(308,-140){\footnotesize{$u_{2\un{1}}$}}
\put(308,-290){\footnotesize{$u_{1\un{2}}$}}
\end{picture}
\end{center}
where $u_{i\un{j}}:=u+c_i-c_{\un{j}}$ for $i,j=1,2$, and $c_1,c_2,c_{\un{1}},c_{\un{2}}\in\C$ are the parameters of the fused solutions. We note that there is no ambiguity in the ordering of the factors, because $R_{1,\un{1}}(u_{1\un{1}})$ and $R_{2,\un{2}}(u_{2\un{2}})$ commute. The expression obtained from the above picture is equal to the right hand side of the defining formula (\ref{def-fus-R}) for $n=n'=2$.
\hfill$\triangle$
\end{remark}

\begin{example}\label{ex1}
In this example, we let $R(u)$ be the Yang solution (\ref{Yang}). We set $n=2$, $n'=1$, $c_1=c_{\un{1}}=0$ and $c_2=1$. We have:
\begin{equation}\label{ex-fus1}R_{\bc,\un{\bc}}(u)=R_{2,\un{1}}(u+1)R_{1,\un{1}}(u)=\left(\textrm{Id}_{V^{\otimes 3}}-\frac{\Pe_{2,\un{1}}}{u+1}\right)\left(\textrm{Id}_{V^{\otimes 3}}-\frac{\Pe_{1,\un{1}}}{u}\right)\ .\end{equation}
We consider the case $\dim(V)=2$ and let $\{e_1,e_2\}$ be a basis of $V$. We define $e_{11}:=e_1\otimes e_1$, $e_{(12)}:=e_1\otimes e_2+e_2\otimes e_1$, $e_{22}:=e_2\otimes e_2$ and $e_{[12]}:=e_1\otimes e_2-e_2\otimes e_1$. We calculate the matrix of the endomorphism (\ref{ex-fus1}) in the following basis of $V^{\otimes 2}\otimes V$:
$$ \{e_{11}\otimes e_1,\,e_{(12)}\otimes e_1,\,e_{22}\otimes e_1,\,e_{11}\otimes e_2,\,e_{(12)}\otimes e_2,\,e_{22}\otimes e_2,\,e_{[12]}\otimes e_1,\,e_{[12]}\otimes e_2\}\ .$$
We obtain (points indicate coefficients equal to $0$):
$$\left(\begin{array}{cccccccc} 
\frac{u-1}{u+1}&\cdot&\cdot&\cdot&\cdot&\cdot&\cdot&\cdot \\
\cdot&\frac{u}{u+1}&\cdot&-\frac{1}{u+1}&\cdot&\cdot&-\frac{1}{u(u+1)}&\cdot \\
\cdot&\cdot&1&\cdot&-\frac{2}{u+1}&\cdot&\cdot&-\frac{2}{u(u+1)} \\
\cdot&-\frac{2}{u+1}&\cdot&1&\cdot&\cdot&\frac{2}{u(u+1)}&\cdot \\
\cdot&\cdot&-\frac{1}{u+1}&\cdot&\frac{u}{u+1}&\cdot&\cdot&\frac{1}{u(u+1)} \\
\cdot&\cdot&\cdot&\cdot&\cdot&\frac{u-1}{u+1}&\cdot&\cdot \\
\cdot&\cdot&\cdot&\cdot&\cdot&\cdot&\frac{u-1}{u}&\cdot \\
\cdot&\cdot&\cdot&\cdot&\cdot&\cdot&\cdot&\frac{u-1}{u} \\
\end{array}\right)$$
We remark that, for the particular choice of $c_1,c_2,c_{\un{1}}$ made here, the fused operator $R_{\bc,\un{\bc}}(u)$ leaves invariant the subspace $S^2V\otimes V$ of $V^{\otimes 2}\otimes V$, where $S^2V$ is the symmetric square of $V$. This phenomenon can be explained by the fact that $S^2V$ is the image of the endomorphism $R(-1)$ of $V^{\otimes 2}$ (since $R(-1)=\textrm{Id}_{V^{\otimes 2}}+\Pe$) together with the following calculation:
$$R_{\bc,\un{\bc}}(u)\cdot R_{1,2}(-1)=R_{2,\un{1}}(u+1)R_{1,\un{1}}(u)\cdot R_{1,2}(-1)=R_{1,2}(-1)\cdot R_{1,\un{1}}(u)R_{2,\un{1}}(u+1)\ ,$$
which shows that $R_{\bc,\un{\bc}}(u)$ leaves invariant the image of $R_{1,2}(-1)$ in $V^{\otimes 2}\otimes V$.

Similarly, if we consider the same situation with only $c_2$ changed to $-1$ instead of $1$, we can verify that the fused operator $R_{\bc,\un{\bc}}(u)$ leaves invariant the subspace $\Lambda^2V\otimes V$ of $V^{\otimes 2}\otimes V$, where $\Lambda^2V$ is the alternating square of $V$. Now, $\Lambda^2V$ is the image of the endomorphism $R(1)$ of $V^{\otimes 2}$ (since $R(1)=\textrm{Id}_{V^{\otimes 2}}-\Pe$) and the invariance of $\Lambda^2V\otimes V$ is explained by
$$R_{\bc,\un{\bc}}(u)\cdot R_{1,2}(1)=R_{2,\un{1}}(u-1)R_{1,\un{1}}(u)\cdot R_{1,2}(1)=R_{1,2}(1)\cdot R_{1,\un{1}}(u)R_{2,\un{1}}(u-1)\ .$$
These two examples are actually (simplest) examples of a general phenomenon, which will be explained in the next section in full generality.
\hfill$\triangle$
\end{example}

\section{Invariant subspaces for fused solutions}\label{sec-inv}

The goal of this Section is to introduce one of the main object of our study: the fusion function $F(\bc)$. We prove some properties of this fusion function and use it to construct distinguished subspaces and show that these subspaces are invariant subspaces for the fused solution. We are still working in a general setting, which is to be applied in the next sections.

\subsection{Invariant subspaces as images of certain operators}\label{subsec-inv}

Let $n$ be a positive integer such that $n\geq 2$ and $\bc:=(c_1,\dots,c_n)$ be an $n$-tuple of complex parameters. Recall that $R(u)$ is an arbitrary solution of the Yang--Baxter equation (\ref{YB-add}) on $V\otimes V$. We define the following endomorphism of $V^{\otimes n}$:
\begin{equation}\label{F}
F(\bc):=\prod_{1\leq i<j\leq n}^{\rightarrow}R_{i,j}(c_i-c_j)\ .
\end{equation}
It will be useful to consider $c_1,\dots,c_n$ as complex variables and to see $F$ as a function of $\bc$ with values in $\End(V^{\otimes n})$. We note that the function $F$ can have singularities and thus the operator $F(\bc)$ can be undefined for some values of $\bc$, depending on the function $R$ (see, for example, (\ref{Yang})).

We will need a preliminary Lemma concerning the operator $F(\bc)$.
\begin{lemma}\label{lem-F}
We have
\begin{equation}\label{F2}
F(\bc)=\prod_{1\leq i<j\leq n}^{\leftarrow}R_{i,j}(c_i-c_j)\ .
\end{equation}
\end{lemma}
\begin{proof}
As in the proof of the Theorem \ref{thm-fus-R}, we do not need to pay attention to the spectral parameters since $(c_i-c_j)+(c_j-c_k)=c_i-c_k$. Thus, until the end of the proof, we use the abbreviated notation $R_{i,j}:=R_{i,j}(c_i-c_j)$ for any $1\leq i<j\leq n$. 

We prove the formula (\ref{F2}) by induction on $n$. For $n=2$ there is nothing to prove, so we let $n>2$ and write
$$F(\bc)=\prod_{1\leq i<j\leq n}^{\rightarrow}R_{i,j}=\Bigl(\prod_{1\leq i<j\leq n-1}^{\rightarrow}R_{i,j}\Bigr)\cdot R_{1,n}R_{2,n}\dots R_{n-1,n}\ ,$$
where we have moved to the right the operators $R_{i,n}$ using commutation relations. Now we use the induction hypothesis and we have
$$\begin{array}{rl}F(\bc) & = \displaystyle\Bigl(\prod_{1\leq i<j\leq n-1}^{\leftarrow}R_{i,j}\Bigr)\cdot R_{1,n}R_{2,n}\dots R_{n-1,n}\\[1.4em]
 & =  \displaystyle\Bigl(\prod_{2\leq i<j\leq n-1}^{\leftarrow}R_{i,j}\Bigr)\cdot R_{1,n-1}\dots R_{1,2}\cdot R_{1,n}R_{2,n}\dots R_{n-1,n}\ .
\end{array}$$
Note that 
$$\begin{array}{rl}\!R_{1,n-1}\dots R_{1,3}\cdot R_{1,2}R_{1,n}R_{2,n}\cdot R_{3,n}\dots R_{n-1,n}\!\!\! & = R_{1,n-1}\dots R_{13}\cdot R_{2,n}R_{1,n}R_{1,2}\cdot R_{3,n}\dots R_{n-1,n}\\[0.2em]
 & = R_{2,n}\cdot R_{1,n-1}\dots R_{1,3}R_{1,n}R_{3,n}\dots R_{n-1,n}\cdot R_{1,2}\,,
 \end{array}$$
and thus, repeating a similar calculation the necessary number of times, we arrive at
$$F(\bc) =  \displaystyle\Bigl(\prod_{2\leq i<j\leq n-1}^{\leftarrow}R_{i,j}\Bigr)R_{2,n}\dots R_{n-1,n}\cdot R_{1,n}R_{1,n-1}\dots R_{1,2}\ .$$
We use again the induction hypothesis and commutation relations to conclude as follows:
$$\begin{array}{rl}F(\bc) & =  \displaystyle\Bigl(\prod_{2\leq i<j\leq n-1}^{\rightarrow}R_{i,j}\Bigr)R_{2,n}\dots R_{n-1,n}\cdot R_{1,n}R_{1,n-1}\dots R_{1,2}\\[1.3em]
& =  \displaystyle\Bigl(\prod_{2\leq i<j\leq n}^{\rightarrow}R_{i,j}\Bigr)R_{1,n}R_{1,n-1}\dots R_{1,2}\\[1.3em]
& =  \displaystyle\Bigl(\prod_{2\leq i<j\leq n}^{\leftarrow}R_{i,j}\Bigr)R_{1,n}R_{1,n-1}\dots R_{1,2}= \displaystyle\prod_{1\leq i<j\leq n}^{\leftarrow}R_{i,j}\ .
\end{array}$$
\end{proof}

For an $n$-tuple of complex numbers $\bc$ such that $F(\bc)$ is a well-defined endomorphism, let $W_{\bc}\subset V^{\otimes n}$ denote the image of $F(\bc)$:
\begin{equation}\label{W}
W_{\bc}:=\textrm{Im}(F(\bc))\subset V^{\otimes n}\ .
\end{equation}
Let $n,n'\geq 2$ be positive integers and $\bc:=(c_1,\dots,c_n)$, respectively $\un{\bc}:=(c_{\un{1}},\dots,c_{\un{n'}})$, be an $n$-tuple, respectively an $n'$-tuple, of complex numbers such that $W_{\bc}$ and $W_{\un{\bc}}$ are defined.
\begin{theorem}\label{thm-inv-sub}
The fused operator $R_{\bc,\un{\bc}}(u)$ preserves the subspace $W_{\bc}\otimes W_{\un{\bc}}$ of $V^{\otimes n}\otimes V^{\otimes n'}$.
\end{theorem}
\begin{proof}
Let $v$ be a complex variable. We first prove the following Lemma:
\begin{lemma}\label{lem-inv-sub}
\begin{enumerate}
\item We consider the space $V\otimes V^{\otimes n'}$ with the copies of $V$ labelled by $\un{0},\un{1},\dots,\un{n'}$ (from left to right). The endomorphism $R_{\un{0},\un{1}}(v-c_{\un{1}})R_{\un{0},\un{2}}(v-c_{\un{2}})\dots R_{\un{0},\un{n'}}(v-c_{\un{n'}})$ preserves the subspace $V\otimes W_{\un{\bc}}\subset V\otimes V^{\otimes n'}$.
\item We consider the space $V^{\otimes n}\otimes V$ with the copies of $V$ labelled by $1,\dots,n,n+1$ (from left to right). The endomorphism $R_{n,n+1}(c_n-v)\dots R_{2,n+1}(c_2-v)R_{1,n+1}(c_1-v)$ preserves the subspace $W_{\bc}\otimes V\subset V^{\otimes n}\otimes V$.
\end{enumerate}
\end{lemma}
\begin{proof}[Proof of the lemma] \begin{enumerate}
\item Set $c_{\un{0}}:=v$. The image of the endomorphism $\textrm{Id}\otimes F(\un{\bc})$ of $V\otimes V^{\otimes n'}$ is $V\otimes W_{\un{\bc}}$. We have, using Lemma \ref{lem-F},
$$\begin{array}{rl} & R_{\un{0},\un{1}}(v-c_{\un{1}})R_{\un{0},\un{2}}(v-c_{\un{2}})\dots R_{\un{0},\un{n'}}(v-c_{\un{n'}})\cdot \bigl(\textrm{Id}\otimes F(\un{\bc})\bigr)\\[0.2em]
= & \displaystyle\prod_{0\leq i<j\leq n'}^{\rightarrow}R_{\un{i},\un{j}}(c_{\un{i}}-c_{\un{j}})
= 
\displaystyle\prod_{0\leq i<j\leq n'}^{\leftarrow}R_{\un{i},\un{j}}(c_{\un{i}}-c_{\un{j}})\\[1.2em]
= & \displaystyle\Bigl(\prod_{1\leq i<j\leq n'}^{\leftarrow}R_{\un{i},\un{j}}(c_{\un{i}}-c_{\un{j}})\Bigr)\cdot R_{\un{0},\un{n'}}(v-c_{\un{n'}})\dots R_{\un{0},\un{2}}(v-c_{\un{2}})R_{\un{0},\un{1}}(v-c_{\un{1}})\\[1.4em]
= & \bigl(\textrm{Id}\otimes F(\un{\bc})\bigr)\cdot R_{\un{0},\un{n'}}(v-c_{\un{n'}})\dots R_{\un{0},\un{2}}(v-c_{\un{2}})R_{\un{0},\un{1}}(v-c_{\un{1}})\ .
\end{array} $$
Thus the operator $R_{\un{0},\un{1}}(v-c_{\un{1}})R_{\un{0},\un{2}}(v-c_{\un{2}})\dots R_{\un{0},\un{n'}}(v-c_{\un{n'}})$ restricted on $V\otimes W_{\un{\bc}}$ has its image contained in $V\otimes W_{\un{\bc}}$.
\item Now set $c_{n+1}:=v$. The image of the endomorphism $F(\bc)\otimes\textrm{Id}$ of $V^{\otimes n}\otimes V$ is $W_{\bc}\otimes V$. We have
$$\begin{array}{rl} & R_{n,n+1}(c_n-v)\dots R_{2,n+1}(c_2-v)R_{1,n+1}(c_1-v)\cdot \bigl( F(\bc)\otimes\textrm{Id}\bigr)\\[0.2em]
= & R_{n,n+1}(c_n-v)\dots R_{2,n+1}(c_2-v)R_{1,n+1}(c_1-v)\cdot\displaystyle\Bigl(\prod_{1\leq i<j\leq n}^{\leftarrow}R_{i,j}(c_{i}-c_{j})\Bigr)\\[-0.5em]
= & \displaystyle\Bigl(\prod_{1\leq i<j\leq n+1}^{\leftarrow}R_{i,j}(c_{i}-c_{j})\Bigr)\\[1.2em]
= & \displaystyle\Bigl(\prod_{1\leq i<j\leq n+1}^{\rightarrow}R_{i,j}(c_{i}-c_{j})\Bigr)\\[1.2em]
= &\displaystyle\Bigl(\prod_{1\leq i<j\leq n}^{\rightarrow}R_{i,j}(c_{i}-c_{j})\Bigr)\cdot R_{1,n+1}(c_1-v)R_{2,n+1}(c_2-v)\dots R_{n,n+1}(c_n-v) \\[1.4em]
= & \bigl( F(\bc)\otimes\textrm{Id}\bigr)\cdot R_{1,n+1}(c_1-v)R_{2,n+1}(c_2-v)\dots R_{n,n+1}(c_n-v)\ .
\end{array} $$
We used Lemma \ref{lem-F} in the first and third equalities, and commutation relations in the second and fourth equalities.
We conclude that the image of the restriction on $W_{\bc}\otimes V$ of the operator $R_{n,n+1}(c_n-v)\dots R_{2,n+1}(c_2-v)R_{1,n+1}(c_1-v)$ is contained in $W_{\bc}\otimes V$.
\end{enumerate}
\vspace{-0.55cm}
\end{proof}
We return to the proof of Theorem \ref{thm-inv-sub}. Recall that the operator $R_{\bc,\un{\bc}}(u)$ is given by Formula (\ref{def-fus-R}), namely
$$ R_{\bc,\un{\bc}}(u):=\prod_{i=1,\dots,n'}^{\rightarrow}R_{n,\un{i}}(u+c_n-c_{\un{i}})\dots\dots R_{2,\un{i}}(u+c_2-c_{\un{i}})R_{1,\un{i}}(u+c_1-c_{\un{i}})\ .$$
For any $i\in\{1,\dots,n'\}$, applying Lemma \ref{lem-inv-sub}(ii) with $v=c_{\un{i}}-u$ and the label $n+1$ replaced by $\un{i}$, we obtain that the operator $R_{n,\un{i}}(u+c_n-c_{\un{i}})\dots\dots R_{2,\un{i}}(u+c_2-c_{\un{i}})R_{1,\un{i}}(u+c_1-c_{\un{i}})$ preserves the subspace $W_{\bc}\otimes V^{\otimes n'}$, and so in turn that the operator $R_{\bc,\un{\bc}}(u)$ preserves $W_{\bc}\otimes V^{\otimes n'}$.

Now recall the alternative formula (\ref{def-fus-R2}) for the operator $R_{\bc,\un{\bc}}(u)$:
$$ R_{\bc,\un{\bc}}(u)=\prod_{i=1,\dots,n}^{\leftarrow}R_{i,\un{1}}(u+c_i-c_{\un{1}})R_{i,\un{2}}(u+c_i-c_{\un{2}})\dots\dots R_{i,\un{n'}}(u+c_i-c_{\un{n'}})\ .$$
For any $i\in\{1,\dots,n\}$, applying Lemma \ref{lem-inv-sub}(i) with $v=u+c_i$ and the label $\un{0}$ replaced by $i$, we obtain that the operator $R_{i,\un{1}}(u+c_i-c_{\un{1}})R_{i,\un{2}}(u+c_i-c_{\un{2}})\dots\dots R_{i,\un{n'}}(u+c_i-c_{\un{n'}})$ preserves the subspace $V^{\otimes n}\otimes W_{\un{\bc}}$, and so in turn that the operator $R_{\bc,\un{\bc}}(u)$ preserves  $V^{\otimes n}\otimes W_{\un{\bc}}$.

We conclude that the operator $R_{\bc,\un{\bc}}(u)$ preserves the subspace $W_{\bc}\otimes W_{\un{\bc}}$.
\end{proof}

\begin{remark}\label{rem-graph-inv}
For $n=n'=2$, part of Theorem \ref{thm-inv-sub} is illustrated as follows, using the pictorial version of the Yang--Baxter equation explained in Remark \ref{rem-graph},
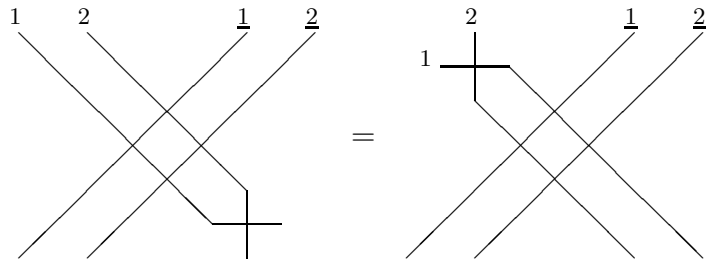
\begin{center}
\setlength{\unitlength}{0.006cm}
\begin{picture}(1500,530)(0,-500)
\put(0,20){\footnotesize{1}}
\put(150,20){\footnotesize{2}}
\put(500,20){\footnotesize{$\un{1}$}}
\put(650,20){\footnotesize{$\un{2}$}}
\put(0,0){ \line(1, -1){425}}
\put(425,-425){ \line(1, 0){150}}
\put(150,0){ \line(1, -1){350}}
\put(500,-350){ \line(0, -1){150}}
\put(500,0){ \line(-1, -1){500}}
\put(650,0){ \line(-1, -1){500}}
\put(750,-250){\large{$=$}}
\put(900,-75){\footnotesize{1}}
\put(1000,20){\footnotesize{2}}
\put(1350,20){\footnotesize{$\un{1}$}}
\put(1500,20){\footnotesize{$\un{2}$}}
\put(1000,-150){ \line(1, -1){350}}
\put(1000,0){ \line(0, -1){150}}
\put(1075,-75){ \line(1, -1){425}}
\put(1075,-75){ \line(-1, 0){150}}
\put(1350,0){ \line(-1, -1){500}}
\put(1500,0){ \line(-1, -1){500}}
\end{picture}
\end{center}
where we omit the spectral parameters. This graphical equality corresponds to the fact that the operator $R_{\bc,\un{\bc}}(u)$ preserves the subspace $W_{\bc}\otimes V^{\otimes n'}$ (a similar picture can be drawn to illustrate the other half of Theorem \ref{thm-inv-sub}).
\hfill$\triangle$
\end{remark}

\begin{remark}
The content of Theorem \ref{thm-inv-sub} is empty if the endomorphisms $F({\bc})$ and $F({\un{\bc}})$ are both invertible (as in this situation $W_{\bc}\otimes W_{\un{\bc}}=V^{\otimes n}\otimes V^{\otimes n'}$). This is the goal of Sections \ref{sec-Sn}--\ref{sec-Hn} to explain that, for standard examples of $R(u)$, there are some values of $\bc$ such that the images of the operators $F(\bc)$ are interesting proper subspaces of $V^{\otimes n}$.\hfill$\triangle$
\end{remark}

\subsection{Alternative formula for the operator $F(\bc)$}\label{subsec-alt}

We will give an alternative formula for $F(\bc)$ which will be useful later. To do this, we define a function $\hR$ with values in $\End(V\otimes V)$ by:
\begin{equation}\label{hR}
\hR(u):=R(u)\Pe\ ,
\end{equation}
where we recall that $\Pe$ is the permutation operator on $V\otimes V$.

For any $\pi$ in the symmetric group $S_n$ on $n$ letters, we define $\Pe_{\pi}\in\End(V^{\otimes n})$ by:
\begin{equation}\label{P-pi}\Pe_{\pi}(x_1\otimes x_2\otimes\dots\otimes x_n):=x_{\pi^{-1}(1)}\otimes x_{\pi^{-1}(2)}\otimes\dots\otimes x_{\pi^{-1}(n)}\ \ \quad\text{for $x_1,\dots,x_n\in V$.}\end{equation}
Note that if $\pi$ is written as a a product of transposition, say $\pi=(i_1,j_1)(i_2,j_2)\dots (i_k,j_k)\in S_n$, then we have with our notation $\Pe_{\pi}=\Pe_{i_1,j_1}\Pe_{i_2,j_2}\dots\Pe_{i_k,j_k}$.

Let $w_n$ be the longest element of the symmetric group $S_n$. We recall the following property of the element $w_n$:
\begin{equation}\label{wn}
w_n=(1,2)(2,3)\dots (n\!-\!1,n)\!\cdot\!w_{n-1}\ \quad\text{and}\ \quad \ w_n(i)=n-i+1\ \ \  \text{for all $i=1,\dots,n$,}
\end{equation}
where $w_{n-1}$ is the longest element of $S_{n-1}$, seen as an element of $S_n$ acting only on the letters $1,\dots,n-1$. Note that $w_n$ is an involution.

\begin{lemma}\label{lem-F-alt}
We have
\begin{equation}\label{F4}
F(\bc)=\Bigl(\prod_{i=1,\dots,n-1}^{\rightarrow}\hR_{i,i+1}(c_1-c_{i+1})\dots\hR_{2,3}(c_{i-1}-c_{i+1})\hR_{1,2}(c_i-c_{i+1})\Bigr)\cdot \Pe_{w_n}\ .
\end{equation}
\end{lemma}
\begin{proof}
We prove Formula (\ref{F4}) by induction on $n$. For $n=2$ there is nothing to prove. Let $n>2$ and write
$$ F(\bc)\!=\!\!\!\!\!\prod_{1\leq i<j\leq n}^{\rightarrow}\!\!\!\!\!R_{i,j}(c_i-c_j)=\Bigl(\!\prod_{1\leq i<j\leq n-1}^{\rightarrow}\!\!\!\!\!\!R_{i,j}(c_i-c_j)\Bigr)\cdot R_{1,n}(c_1-c_n)R_{2,n}(c_2-c_n)\dots R_{n-1,n}(c_{n-1}-c_n)\,.$$
The induction hypothesis allows to replace $\displaystyle\!\prod_{1\leq i<j\leq n-1}^{\rightarrow}\!\!\!R_{i,j}(c_i-c_j)$ by
$$
\Bigl(\prod_{i=1,\dots,n-2}^{\rightarrow}\hR_{i,i+1}(c_1-c_{i+1})\dots\hR_{2,3}(c_{i-1}-c_{i+1})\hR_{1,2}(c_i-c_{i+1})\Bigr)\cdot \Pe_{w_{n-1}}\ .
$$
So it remains to prove that
\begin{equation}\label{eq-lem1}
\begin{array}{l}\Pe_{w_{n-1}}R_{1,n}(c_1-c_n)R_{2,n}(c_2-c_n)\dots R_{n-1,n}(c_{n-1}-c_n)\\[0.2em]
\hspace{4.5cm}=\hR_{n-1,n}(c_1-c_{n})\dots\hR_{2,3}(c_{n-2}-c_{n})\hR_{1,2}(c_{n-1}-c_{n})\Pe_{w_n}\,.
\end{array}
\end{equation}
Using (\ref{wn}), the left hand side of (\ref{eq-lem1}) is equal to
$$\begin{array}{rl}
& R_{n-1,n}(c_1-c_n)R_{n-2,n}(c_2-c_n)\dots\dots R_{1,n}(c_{n-1}-c_n)\cdot\Pe_{w_{n-1}}\\[0.2em]
= & \hR_{n-1,n}(c_1-c_n)\Pe_{n-1,n}\dots\dots \hR_{2,n}(c_{n-2}-c_n)\Pe_{2,n}\,\hR_{1,n}(c_{n-1}-c_n)\Pe_{1,n}\cdot\Pe_{w_{n-1}}\ .
\end{array}$$
Moving all the permutation operators to the right, we obtain that the left hand side of (\ref{eq-lem1}) is equal to
$$\hR_{n-1,n}(c_1-c_{n})\dots\hR_{2,3}(c_{n-2}-c_{n})\hR_{1,2}(c_{n-1}-c_{n})\Pe_{n-1,n}\Pe_{n-2,n}\dots\Pe_{1,n}\cdot\Pe_{w_{n-1}}\ .$$
It is easy to see that $\Pe_{n-1,n}\Pe_{n-2,n}\dots\Pe_{1,n}=\Pe_{1,2}\Pe_{2,3}\dots\Pe_{n-1,n}$ and thus, with (\ref{wn}), this concludes the proof of (\ref{eq-lem1}) and in turn of the lemma.
\end{proof}

We define the following function of $\bc$ with values in $\End(V^{\otimes n})$:
\begin{equation}\label{hF}
\hF(\bc):=\prod_{i=1,\dots,n-1}^{\rightarrow}\hR_{i,i+1}(c_1-c_{i+1})\dots\hR_{2,3}(c_{i-1}-c_{i+1})\hR_{1,2}(c_i-c_{i+1})\ .
\end{equation}
Let $\bc$ be such that $F(\bc)$ is a well-defined endomorphism, so that the space $W_{\bc}$ is defined. Lemma \ref{lem-F-alt} implies that $\hF(\bc):=F(\bc)\Pe_{w_n}$ and thus has the following corollary. 
\begin{corollary}\label{cor-hF}
The subspace $W_{\bc}\subset V^{\otimes n}$ coincides with the image of the endomorphism $\hF(\bc)$.
\end{corollary}

\subsection{Admissible permutation of the parameters $(c_1,\dots,c_n)$}\label{subsec-dep}

For any permutation $\pi$ in the symmetric group $S_n$, we define $\bc^{(\pi)}:=(c_{\pi^{-1}(1)},\dots,c_{\pi^{-1}(n)})$. We denote by $s_k$ the transposition $(k,k+1)\in S_n$, for $k=1,\dots,n-1$; then we have, for $k=1,\dots,n-1$, $\bc^{(s_k)}=(c_1,\dots,c_{k+1},c_k,\dots,c_n)$.

\begin{lemma}\label{lem-dep-R}
\begin{enumerate}
\item For $k=1,\dots,n-1$, we have, on $V^{\otimes n}\otimes V^{\otimes n'}$,
\begin{equation}\label{dep-R}
\Pe_{k,k+1}R_{k+1,k}(c_{k+1}-c_k)\cdot R_{\bc,\un{\bc}^{\phantom{(s_k)}}\!\!\!\!\!\!\!\!\!\!}(u)=R_{\bc^{(s_k)}\!,\un{\bc}}(u)\cdot \Pe_{k,k+1}R_{k+1,k}(c_{k+1}-c_k)\ .
\end{equation}
\item For $k=1,\dots,n-1$, we have, on $V^{\otimes n}$,
\begin{equation}\label{dep-F}
\Pe_{k,k+1}R_{k+1,k}(c_{k+1}-c_k)\cdot F(\bc)=F(\bc^{(s_k)})\cdot \Pe_{k,k+1}R_{k,k+1}(c_{k}-c_{k+1})\ .
\end{equation}
\end{enumerate}
\end{lemma}
\begin{proof}
(i) We have, for any $i\in\{1,\dots,n'\}$,
\[\begin{array}{ll}
& \Pe_{k,k+1}R_{k+1,k}(c_{k+1}-c_k)\cdot \prod\limits_{j=1,\dots,n}^{\leftarrow}R_{j,\un{i}}(u+c_j-c_{\un{i}})\\
= & \Pe_{k,k+1}\cdot \Bigl(\prod\limits_{j=1,\dots,n}^{\leftarrow}R_{s_k(j),\un{i}}(u+c_{s_k(j)}-c_{\un{i}})\Bigr) \cdot R_{k+1,k}(c_{k+1}-c_k) \\
= & \Bigl(\prod\limits_{j=1,\dots,n}^{\leftarrow}R_{j,\un{i}}(u+c_{s_k(j)}-c_{\un{i}})\Bigr) \cdot\Pe_{k,k+1}R_{k+1,k}(c_{k+1}-c_k)\ ;
\end{array}\]
we use in the second equality commutation relations, together with the Yang--Baxter relation on the copies labelled by $k+1$, $k$ and $\un{i}$. Due to Formula (\ref{def-fus-R}), this proves the item (i).

(ii) Formula (\ref{dep-F}) is equivalent to
\begin{equation}\label{dep-F2}
R_{k+1,k}(c_{k+1}-c_k)\cdot F(\bc)=\Bigl(\prod\limits_{1\leq i<j\leq n}^{\rightarrow}R_{s_k(i),s_k(j)}(c_{s_k(i)}-c_{s_k(j)})\Bigr)\cdot R_{k,k+1}(c_{k}-c_{k+1})\ .
\end{equation}
We prove Formula (\ref{dep-F2}) by induction on $n$. If $n=2$, Formula (\ref{dep-F2}) is trivial.

Let $n>2$. We deal first with the case $k=1$. We have, using commutation relations,
\[F(\bc) = \Bigl(\prod\limits_{1\leq i<j\leq n-1}^{\rightarrow}R_{i,j}(c_{i}-c_{j})\Bigr)\cdot R_{1,n}(c_1-c_n)R_{2,n}(c_2-c_n)\dots R_{n-1,n}(c_{n-1}-c_n)\ .\]
The induction hypothesis gives
$$R_{2,1}(c_2-c_1)\cdot \Bigl(\prod\limits_{1\leq i<j\leq n-1}^{\rightarrow}R_{i,j}(c_{i}-c_{j})\Bigr)= \Bigl(\prod\limits_{1\leq i<j\leq n-1}^{\rightarrow}R_{s_1(i),s_1(j)}(c_{s_1(i)}-c_{s_1(j)})\Bigr)\cdot R_{1,2}(c_1-c_2)\ ,$$
and besides we have
$$\begin{array}{ll}
 & R_{1,2}(c_1-c_2)\cdot R_{1,n}(c_1-c_n)R_{2,n}(c_2-c_n)R_{3,n}(c_3-c_n)\dots R_{n-1,n}(c_{n-1}-c_n)\\[0.4em]
= & R_{2,n}(c_2-c_n)R_{1,n}(c_1-c_n)R_{3,n}(c_3-c_n)\dots R_{n-1,n}(c_{n-1}-c_n)\cdot R_{1,2}(c_1-c_2)\ ,
\end{array}$$
where we used the Yang--Baxter relation on the copies labelled by $1$, $2$ and $n$, and commutation relations. Formula (\ref{dep-F2}) for $k=1$ follows.

Then we assume that $k>1$, and we write 
\[F(\bc) = R_{1,2}(c_1-c_2)R_{1,3}(c_1-c_3)\dots R_{1,n}(c_{1}-c_n)\cdot\Bigl(\prod\limits_{2\leq i<j\leq n}^{\rightarrow}R_{i,j}(c_{i}-c_{j})\Bigr)\ .\]
Similarly to above, we first have, using commutation relations and the Yang--Baxter equation on the copies labeled by $1$, $k+1$ and $k$, that
\[\begin{array}{ll}
 & R_{k+1,k}(c_{k+1}-c_k)\cdot R_{1,2}(c_1-c_2)\dots R_{1,k}(c_1-c_k)R_{1,k+1}(c_1-c_{k+1})\dots R_{1,n}(c_{1}-c_n)\\[0.4em]
= & R_{1,2}(c_1-c_2)\dots R_{1,k+1}(c_1-c_{k+1})R_{1,k}(c_1-c_{k})\dots R_{1,n}(c_{1}-c_n)\cdot R_{k+1,k}(c_{k+1}-c_k)\ ,
\end{array}\]
and then we use the induction hypothesis, namely (since $k>1$)
\[R_{k+1,k}(c_{k+1}-c_k)\cdot\Bigl(\prod\limits_{2\leq i<j\leq n}^{\rightarrow}R_{i,j}(c_{i}-c_{j})\Bigr)=\Bigl(\prod\limits_{2\leq i<j\leq n}^{\rightarrow}R_{s_k(i),s_k(j)}(c_{s_k(i)}-c_{s_k(j)})\Bigr)\cdot R_{k,k+1}(c_{k+1}-c_k)\ .\]
This yields the desired result.
\end{proof}

For the remaining of this section, we assume that
\begin{equation}\label{inv-R}
R(u)R_{2,1}(-u)=\gamma(u)\cdot \textrm{Id}_{V^{\otimes 2}}\ \ \ \ \ \ \text{for a complex-valued function $\gamma$ of $u$\,,}
\end{equation}
where $R_{2,1}(u):=\Pe R(u)\Pe$. In the physical literature, this condition is called the ``unitarity'' condition. We note that $\gamma(-u)=\gamma(u)$.
\begin{definition}\label{def-adm}
For $k\in\{1,\dots,n-1\}$, we say that the transposition $s_k=(k,k+1)$ is \emph{admissible} for $R(u)$ and for the set of parameters $\bc=(c_1,\dots,c_n)$ if $\gamma(c_k-c_{k+1})\neq 0$.
\end{definition}

Let $k\in\{1,\dots,n-1\}$ such that the transposition $s_k$ is \emph{admissible} for $R(u)$ and for the set of parameters $\bc=(c_1,\dots,c_n)$. We define an endomorphism $A_k$ of $V^{\otimes n}$ by
\begin{equation}\label{def-Ak}
A_k:=\Pe_{k,k+1}R_{k+1,k}(c_{k+1}-c_k)\ .
\end{equation}
The admissibility of $s_k$ implies that the operator $A_k$ is invertible with, from (\ref{inv-R}), 
$$A_k^{-1}=\gamma_k^{-1}R_{k,k+1}(c_{k}-c_{k+1})\Pe_{k,k+1}\,,\ \ \ \ \text{where $\gamma_k:=\gamma(c_k-c_{k+1})$\ .}$$
In this situation, the assertions of Lemma \ref{lem-dep-R} can be rewritten, with $\widetilde{A}_k:=\Pe_{k,k+1}A_k\Pe_{k,k+1}$,
\begin{equation}\label{A-R-F}R_{\bc^{(s_k)}\!,\un{\bc}}(u)=(A_k\otimes\textrm{Id}_{V^{\otimes n'}})R_{\bc,\un{\bc}^{\phantom{(s_k)}}\!\!\!\!\!\!\!\!\!\!}(u)(A_k^{-1}\otimes\textrm{Id}_{V^{\otimes n'}})\ \ \ \ \text{and}\ \ \ \  A_k F(\bc)=\gamma_k F(\bc^{(s_k)})\widetilde{A}_k^{-1}.\end{equation}
We sum up the consequences of these two formulas in the following proposition. We assume that $W_{\bc}$ is defined (and thus $W_{\bc^{(s_k)}}$ is defined as well from the second relation in (\ref{A-R-F})).
\begin{proposition}\label{prop-inv}
Let $k\in\{1,\dots,n-1\}$ such that the transposition $s_k$ is \emph{admissible} for $R(u)$ and for the set of parameters $\bc=(c_1,\dots,c_n)$. Let $e_1,\dots,e_l\in V^{\otimes n}$ such that $\mathcal{B}_{\bc}:=\{F(\bc)e_1,\dots,F(\bc)e_l\}$ is a basis of the subspace $W_{\bc}$.
\begin{enumerate}
\item Then $\mathcal{B}_{\bc^{(s_k)}}:=\{A_kF(\bc)e_1,\dots,A_kF(\bc)e_l\}$ is a basis of $W_{\bc^{(s_k)}}$. In particular, $W_{\bc^{\phantom{(s_k)}}\!\!\!\!\!\!\!\!\!\!}$ and $W_{\bc^{(s_k)}}$ have the same dimension.
\item Moreover, for any basis $\mathcal{B}_{\un{\bc}}$ of $W_{\un{\bc}}$, the matrix of the endomorphism $R_{\bc,\un{\bc}}(u)\Bigr\vert_{W_{\bc}\otimes W_{\un{\bc}}}$ in the basis $\mathcal{B}_{\bc}\otimes\mathcal{B}_{\un{\bc}}$ coincides with the matrix of the endomorphism $R_{\bc^{(s_k)}\!,\un{\bc}}(u)\Bigr\vert_{W_{\bc^{(s_k)}}\otimes W_{\un{\bc}^{\phantom{(s_k)}}\!\!\!\!\!\!\!\!\!\!}}$ in the basis $\mathcal{B}_{\bc^{(s_k)}}\otimes\mathcal{B}_{\un{\bc}^{\phantom{(s_k)}}\!\!\!\!\!\!\!\!\!\!}$ (we use the notation $\mathcal{B}_{\bc}\otimes\mathcal{B}_{\un{\bc}}:=\{a\otimes b\ |\ a\in\mathcal{B}_{\bc},\ b\in\mathcal{B}_{\un{\bc}} \}^{\phantom{A}}\!\!\!\!$ and similarly for $\mathcal{B}_{\bc^{(s_k)}}\otimes\mathcal{B}_{\un{\bc}^{\phantom{(s_k)}}\!\!\!\!\!\!\!\!\!\!}$).
\end{enumerate}
\end{proposition}
\begin{proof}
\begin{enumerate}
\item As the operator $A_k$ is invertible, the linear independence of the vectors in $\mathcal{B}_{\bc^{(s_k)}}$ is immediate. Moreover, due to the second relation in (\ref{A-R-F}), we have that the dimensions of $W_{\bc^{\phantom{(s_k)}}\!\!\!\!\!\!\!\!\!\!}$ and $W_{\bc^{(s_k)}}$ coincide and also that the vectors in $\mathcal{B}_{\bc^{(s_k)}}$ belong to $W_{\bc^{(s_k)}}$. This proves the item (i).
\item The item (ii) is a direct consequence of the first formula in (\ref{A-R-F}).
\end{enumerate}
\vskip -0.4cm
\end{proof}

\vskip .2cm
\begin{remark}\label{rem-mult}
In Sections \ref{sec-fus} and \ref{sec-inv}, we worked with a solution of the Yang--Baxter equation with ``additive" spectral parameters (see paragraph \textbf{4} in Section \ref{sec-prel}). 
The whole construction in Sections \ref{sec-fus} and \ref{sec-inv} has an equivalent ``multiplicative" version, where one replaces addition for the spectral parameters by multiplication. As examples, the fused solution (\ref{def-fus-R}) is given by
\begin{equation}\label{def-fus-R-mult}
R_{\bc,\un{\bc}}(\alpha):=\prod_{i=1,\dots,n'}^{\rightarrow}R_{n,\un{i}}(\frac{\alpha c_n}{c_{\un{i}}})\dots\dots R_{2,\un{i}}(\frac{\alpha c_2}{c_{\un{i}}})R_{1,\un{i}}(\frac{\alpha c_1}{c_{\un{i}}})\ ,
\end{equation}
the invariant subspace for this fused solution is the image in $V^{\otimes n}$ of the operator
\begin{equation}\label{F-mult}
F(\bc):=\prod_{1\leq i<j\leq n}^{\rightarrow}R_{i,j}(\frac{c_i}{c_j})\ ,
\end{equation}
and the unitarity condition (\ref{inv-R}) becomes $R(\alpha)R_{2,1}(\frac{1}{\alpha})=\gamma(\alpha)\cdot \textrm{Id}_{V^{\otimes 2}}$, with $\gamma(\alpha)=\gamma(\frac{1}{\alpha})\in\C$. All proofs and calculations made in Sections \ref{sec-fus} and \ref{sec-inv} are completely analogous for the multiplicative version. In the sequel we will use both versions. \hfill$\triangle$
\end{remark}

\section{Invariant subspaces and representations of $U(\mathfrak{gl}_N)$}\label{sec-Sn}

In this Section, we will further study the subspaces $W_{\bc}$, when we start from the Yang solution, (\ref{Yang2}) below, of the Yang--Baxter equation on $V\otimes V$. This will lead to a family of fused solutions acting on irreducible representations of the general linear Lie algebra $\mathfrak{gl}_N$.

\vskip .2cm
We fix in this Section:
\begin{equation}\label{Yang2}R(u):=\textrm{Id}_{V^{\otimes 2}}-\frac{\Pe}{u}\ .
\end{equation}

Let $n>1$ and recall that $F$ is the following function of $\bc:=(c_1,\dots,c_n)\in\C^n$ with values in $\End(V^{\otimes n})$:
\begin{equation}\label{F3}
F(\bc):=\prod_{1\leq i<j\leq n}^{\rightarrow}R_{i,j}(c_i-c_j)\ .
\end{equation}
Moreover, recall that the subspace we are interested in, denoted by $W_{\bc}$, is the image of the operator $F(\bc)$, for points $\bc=(c_1,\dots,c_n)$ where the function $F$ is a defined element of $\End(V^{\otimes n})$. For the solution (\ref{Yang2}), the expression for the operator $F(\bc)$ reads:
\begin{equation}\label{F-Sn}F(\bc):=\prod_{1\leq i<j\leq n}^{\rightarrow}(\text{Id}_{V^{\otimes n}}-\frac{\Pe_{i,j}}{c_i-c_j})\ .
\end{equation}

\subsection{Schur--Weyl duality}

Set $N:=\dim(V)$ and denote by $\mathfrak{gl}_N$ the Lie algebra of endomorphisms of $V$ and by $U(\mathfrak{gl}_N)$ its universal enveloping algebra. The  Lie algebra $\mathfrak{gl}_N$ acts on the tensor product $V^{\otimes n}$ and we denote by $\chi$ this representation, given by:
$$ \chi(g):=g\otimes \textrm{Id} \otimes\dots\otimes \textrm{Id}+\textrm{Id} \otimes g\otimes \textrm{Id}\dots\otimes \textrm{Id}+\cdots\dots+\textrm{Id} \otimes\dots\otimes \textrm{Id}\otimes g\ ,\ \quad\text{$g\in\mathfrak{gl}_N$.}$$
The representation $\chi$ extends to a representation, which we still denote by $\chi$, of $U(\mathfrak{gl}_N)$.

The symmetric group $S_n$ acts on $V^{\otimes n}$ by permuting the copies and we denote this representation of $S_n$ by $\rho$. Explicitly, the representation $\rho$ is given by:
\begin{equation}\label{SW-Sn} \rho(\pi)=\Pe_{\pi}\ \quad\text{ for all $\pi\in S_n$,}
\end{equation}
where $\Pe_{\pi}$ is defined by (\ref{P-pi}). The representation $\rho$ extends to a representation, which we still denote by $\rho$, of the group algebra $\C S_n$.

The irreducible complex representations of $\C S_n$ are parametrized by the partitions of $n$. For any partition $\lambda$ of $n$, we denote by $M_{\lambda}^{\C S_n}$ the corresponding irreducible $\C S_n$-module (with the convention that $M_{(n)}^{\C S_n}$ is the trivial representation).

For any partition $\lambda$ such that $\ell(\lambda)\leq N$, we denote by $M_{\lambda}^{U(\mathfrak{gl}_N)}$ the irreducible highest-weight $U(\mathfrak{gl}_N)$-module of highest weight $\lambda$.

The Schur--Weyl duality between the symmetric group $S_n$ and the algebra $U(\mathfrak{gl}_N)$ can be expressed by the following assertions (see \emph{e.g.} \cite{GW,We}).
\begin{theorem}\label{thm-SW-Sn}
\begin{enumerate}
\item The subalgebra $\rho\bigl(\C S_n\bigr)$ of $\text{End}(V^{\otimes n})$ is the centraliser of $\chi\bigl(U(\mathfrak{gl}_N))$.
\item The subalgebra $\chi\bigl(U(\mathfrak{gl}_N))$ of $\text{End}(V^{\otimes n})$ is the centraliser of $\rho\bigl(\C S_n\bigr)$.
\item As an $\left(U(\mathfrak{gl}_N)\!\otimes\!\C S_n\right)$-module defined by $\chi$ and $\rho$, the space $V^{\otimes n}$ decomposes as:
\begin{equation}\label{SW-Sn-for}
V^{\otimes n}\cong \bigoplus_{\lambda} M_{\lambda}^{U(\mathfrak{gl}_N)}\otimes M_{\lambda}^{\C S_n}\ ,
\end{equation}
where $\lambda$ runs over the set of partitions such that $|\lambda|=n$ and $\ell(\lambda)\leq N$.
\end{enumerate}
\end{theorem}

\subsection{Fusion formula for idempotents of the symmetric group}

We consider the following rational function in $u_1,\dots,u_n$ with values in the group algebra $\C S_n$:
\begin{equation}\label{phi-sym}
\Phi(u_1,\dots,u_n):=\prod_{1\leq i<j\leq n}^{\rightarrow}(1-\frac{(i,j)}{u_i-u_j})\ .
\end{equation}
Our interest in $\Phi$ comes from the fact that $\rho\bigl(\Phi(c_1,\dots,c_n)\bigr)=F(\bc)$, according to Formulas (\ref{F-Sn}) and (\ref{SW-Sn}).

We refer to the rational function $\Phi$ as the ``fusion function'' of the symmetric group. It is easy to check that, for distinct $i,j,k\in\{1,\dots,n\}$,
\begin{equation}\label{ijk}(1-\frac{(i,j)}{u})(1-\frac{(i,k)}{u+v})(1-\frac{(j,k)}{v})=(1-\frac{(j,k)}{v})(1-\frac{(i,k)}{u+v})(1-\frac{(i,j)}{u})\ .
\end{equation}
The functions $(1-\frac{(i,j)}{u})$ with values in $\C S_n$ are called \emph{Baxterized elements}. We say that they are ``universal" solutions of the Yang--Baxter equation (associated to the symmetric group) as they provide solutions of the Yang--Baxter equation in the representations of $S_n$.

Let $\lambda$ be a partition of $n$ and let $\mathcal{T}$ be a standard Young tableau of shape $\lambda$. For brevity, set $\cc_i:=\cc(\mathcal{T}|i)$ for $i=1,\dots,n$.
The following result traces back to the work of Jucys \cite{Ju}, see also \cite{Ch2,Mo,Na,Na2,Na3} (we use a formulation as in \cite{Mo}). Recall that $f(\lambda)$ is the non-zero complex number defined in (\ref{f-lam}).
\begin{theorem}\label{thm-idem-Sn}
The element obtained by the following consecutive evaluations (from $u_1$ to $u_n$)
\begin{equation}\label{ET}
E_{\mathcal{T}}:=
f(\lambda)\Phi(u_1,\dots,u_n)\Bigr\vert_{u_1=\cc_1}\Bigr\vert_{u_2=\cc_2}\dots \Bigr\vert_{u_n=\cc_n}
\end{equation}
is a primitive idempotent of $\C S_n$ which generates a minimal left ideal isomorphic, as an $\C S_n$-module, to the irreducible module $M_{\lambda}^{\C S_n}$.
\end{theorem}
An equivalent formulation of this theorem is as follows. We consider the isomorphism of algebras $\C S_n\cong \bigoplus_{\lambda'} \text{End}_{\C}(M_{\lambda'}^{\C S_n})$ where an element of $\C S_n$ is sent in each summand to its image as an operator on $M_{\lambda'}^{\C S_n}$. Then Theorem \ref{thm-idem-Sn} asserts that $E_{\mathcal{T}}$ is sent through this isomorphism to a projector of rank one (\emph{i.e.} a primitive idempotent) in $\text{End}_{\C}(M_{\lambda}^{\C S_n})$, where $\lambda=\sh_{\cT}$, and to $0$ in any other $\text{End}_{\C}(M_{\lambda'}^{\C S_n})$. In other words again, we have
\[\dim\bigl(E_{\mathcal{T}}(M_{\sh_{\cT}}^{\C S_n})\bigr)=1\ \ \ \ \ \text{and}\ \ \ \ \ E_{\mathcal{T}}(M_{\lambda'}^{\C S_n})=0\ \ \ \text{if $\lambda'\neq\sh_{\cT}$\ .}\]

\begin{remark}
The assertion formulated in Theorem \ref{thm-idem-Sn} is what we will need later on. Nevertheless more can be said (see for example \cite{Mo}). In fact, the set $\{E_{\mathcal{T}}\}$, where $\cT$ runs through the set of standard Young tableaux of size $n$, is a complete set of primitive orthogonal idempotents of $\C S_n$. Moreover, the one-dimensional subspace $E_{\mathcal{T}}(M_{\sh_{\cT}}^{\C S_n})$ is generated by the basis vector associated to $\cT$ in the so-called seminormal basis of $M_{\sh_{\cT}}^{\C S_n}$.
\hfill$\triangle$
\end{remark}
\begin{remark}\label{rem-sing}
For some standard Young tableau $\mathcal{T}$, it may happen that $\cc(\mathcal{T}|j)=\cc(\mathcal{T}|k)$ for $j\neq k$ (for example, if $\sh_{\cT}=(2,2)$). In this situation, Theorem \ref{thm-idem-Sn} asserts in particular that the rational function $\Phi(u_1,\dots,u_n)$ gives a well-defined element of $\C S_n$ when evaluated as in (\ref{ET}), even if the factor $\displaystyle(1-\frac{(j,k)}{u_j-u_k})$ is singular at $u_j=\cc(\mathcal{T}|j)$ and $u_k=\cc(\mathcal{T}|k)$.\hfill$\triangle$
\end{remark}

\subsection{Fused solutions on $M_{\lambda}^{U(\mathfrak{gl}_N)}\otimes M_{\lambda'}^{U(\mathfrak{gl}_N)}$}\label{subsec-conc-Sn}

Let $\mathcal{T}$ be a standard Young tableau such that $\sh_{\cT}=\lambda$ with $|\lambda|=n$ and $\ell(\lambda)\leq N$. We let
\begin{equation}\label{def-FT} F(\cT):=F(\bc)\Bigr\vert_{c_1=\cc(\mathcal{T}|1)}\Bigr\vert_{c_2=\cc(\mathcal{T}|2)}\dots \Bigr\vert_{c_n=\cc(\mathcal{T}|n)}\ ,\end{equation}
where $\bc=(c_1,\dots,c_n)$ is seen here as an $n$-tuple of variables. From the fact that $F(\bc)=\rho\bigl(\Phi(c_1,\dots,c_n)\bigr)$, together with (\ref{ET}), we have that
\begin{equation}\label{F-E}F(\cT)=\rho\left(f(\lambda)^{-1} E_{\mathcal{T}}\right)\ .\end{equation}
In particular $F(\cT)$ is a well-defined endomorphism of $V^{\otimes n}$. We denote by $W_{\cT}$ its image. Due to (\ref{F-E}), the space $W_{\cT}$ coincides with the image in $V^{\otimes n}$ of the operator $\rho\left(E_{\mathcal{T}}\right)$.

The preceding discussion, which follows from Theorem \ref{thm-idem-Sn} together with the Schur--Weyl duality, namely Formula (\ref{SW-Sn-for}) in Theorem \ref{thm-SW-Sn}, leads to the identification of the space $W_{\cT}$. 
\begin{theorem}\label{thm-conc-Sn}
The subspace $W_{\cT}$ of $\,V^{\otimes n}$ is an irreducible $U(\mathfrak{gl}_N)$-module isomorphic to $M_{\lambda}^{U(\mathfrak{gl}_N)}\!$.
\end{theorem}

Let $\mathcal{T}$, $\mathcal{T}'$ be two standard Young tableaux such that $\sh_{\cT}=\lambda$ and $\sh_{\cT'}=\lambda'$, with $|\lambda|=n$, $|\lambda'|=n'$ and $\ell(\lambda),\ell(\lambda')\leq N$. We set $\bc^{\cT}\!\!:=\bigl(\cc(\mathcal{T}|1),\dots,\cc(\mathcal{T}|n)\bigr)$, $\un{\bc}^{\cT'}\!\!:=\bigl(\cc(\mathcal{T'}|1),\dots,\cc(\mathcal{T'}|n')\bigr)$ and we define 
$$R^{\text{res}}_{\mathcal{T},\mathcal{T}'}(u):=R_{\bc^{\cT},\un{\bc}^{\cT'}}(u)\Bigr\vert_{W_{\cT}\otimes W_{\cT'}}\ ,$$
the restriction of the operator $R_{\bc^{\cT},\un{\bc}^{\cT'}}(u)$ to the subspace $W_{\cT}\otimes W_{\cT'}$ of $V^{\otimes n}\otimes V^{\otimes n'}$.

Let $k\in\{1,\dots,n-1\}$ and denote by $\mathcal{T}^{(s_k)}$ the Young tableau of shape $\lambda$ obtained from $\mathcal{T}$ by exchanging the nodes with numbers $k$ and $k+1$. We assume that $\mathcal{T}^{(s_k)}$ is also standard. This is equivalent to say that $\cc(\cT|k+1)\neq \cc(\cT|k)\pm1$. 

Moreover, the Yang solution (\ref{Yang2}) satisfies the following unitarity condition:
\begin{equation}\label{inv-R-Sn}R(u)R_{2,1}(-u)=\frac{u^2-1}{u^2}\cdot\textrm{Id}_{V^{\otimes 2}}\,. \end{equation}
Therefore, according to Definition \ref{def-adm}, the condition for the tableau $\mathcal{T}^{(s_k)}$ to be standard is equivalent to the condition for $s_k=(k,k+1)$ to be an admissible transposition for theYang solution $R(u)$ and the set of parameters $\bc^{\cT}$. Theorem \ref{thm-conc-Sn} implies that $W_{\cT}\otimes W_{\cT'}$ and $W_{\cT^{(s_k)}}\otimes W_{\cT'^{\phantom{(s_k)}}\!\!\!\!\!\!\!\!\!}$ are both isomorphic to $M_{\lambda}^{U(\mathfrak{gl}_N)}\otimes M_{\lambda'}^{U(\mathfrak{gl}_N)}$ and, moreover, Proposition \ref{prop-inv} applied here asserts that the endomorphisms $R^{\text{res}}_{\cT,\cT'}(u)$ and $R^{\text{res}}_{\cT^{(s_k)},\cT'}(u)$ coincide up to a change of basis. 

Besides, it is well-known that, for two standard Young tableaux $\mathcal{T}$ and $\widetilde{\mathcal{T}}$ of the same shape $\lambda$, there is a sequence of admissible transpositions $\pi=s_{i_1}\dots s_{i_l}\in S_n$ transforming $\mathcal{T}$ into $\widetilde{\mathcal{T}}$. 

Finally, we sum up the results obtained in this section for the Yang solution. The item (i) below follows from Theorems \ref{thm-fus-R}, \ref{thm-inv-sub} and \ref{thm-conc-Sn}. The item (ii) is a consequence of the above discussion together with the Proposition \ref{prop-inv}.
\begin{corollary}\label{coro-fin-Sn}
\begin{enumerate}
\item The set of functions $\{R^{\text{res}}_{\cT,\cT'}\}$, where $\mathcal{T}$, $\mathcal{T}'$ are standard Young tableaux such that $\ell(\sh_{\cT}),\ell(\sh_{\cT'})\leq N$, forms a family of solutions of the Yang--Baxter equation, where $R^{\text{res}}_{\cT,\cT'}(u)$ is an endomorphism of a space isomorphic to $M_{\sh_{\cT}}^{U(\mathfrak{gl}_N)}\otimes M_{\sh_{\cT'}}^{U(\mathfrak{gl}_N)}$.
\item For four standard Young tableaux $\mathcal{T}$, $\widetilde{\mathcal{T}}$,  $\mathcal{T}'$ and $\widetilde{\mathcal{T}}'$ as above such that $\sh_{\phantom{\widetilde{T}}\!\!\!\!\cT}\!=\sh_{\widetilde{\cT}}$ and $\sh_{\phantom{\widetilde{T}}\!\!\!\!\cT'}\!=\sh_{\widetilde{\cT}'}$, the endomorphisms $R^{\text{res}}_{\cT,\cT'}(u)$ and $R^{\text{res}}_{\widetilde{\cT},\widetilde{\cT}'}(u)$ coincide up to a change of basis.
\end{enumerate}
\end{corollary}

\vskip .2cm
\begin{example}\label{ex-fus-Sn}
Let $\cT=\cT'=\fbox{\scriptsize{1}}\fbox{\scriptsize{2}}\,$. 
The expression for the operator $R_{\cT,\cT'}(u)$ is
$$R_{\cT,\cT'}(u)=(\textrm{Id}_{V^{\otimes 4}}-\frac{\Pe_{2,\un{1}}}{u+1})(\textrm{Id}_{V^{\otimes 4}}-\frac{\Pe_{1,\un{1}}}{u})(\textrm{Id}_{V^{\otimes 4}}-\frac{\Pe_{2,\un{2}}}{u})(\textrm{Id}_{V^{\otimes 4}}-\frac{\Pe_{1,\un{2}}}{u-1})\ .$$

The results of this section, applied to this example, give that the operator $R_{\cT,\cT'}(u)$ preserves the subspace $W_{\Box\!\Box}\otimes W_{\Box\!\Box}\subset V^{\otimes 2}\otimes V^{\otimes 2}$, where $W_{\Box\!\Box}$ is the space of the irreducible representation of $U(\mathfrak{gl}_N)$ corresponding to the partition $\lambda=(2)$. In fact the subspace $W_{\Box\!\Box}\subset V^{\otimes 2}$ coincides with the symmetric square of $V$, since it is the image of the operator $R(-1)=\textrm{Id}_{V^{\otimes 2}}+\Pe$.

Let $N=2$ and fix a basis $\{e_1,e_2\}$ of $V$. We set $\tilde{e}_1:=e_1\otimes e_1$, $\tilde{e}_2:=e_1\otimes e_2+e_2\otimes e_1$ and $\tilde{e}_3:=e_2\otimes e_2$. We give the matrix of the endomorphism $R^{\text{res}}_{\cT,\cT'}(u)$ of $W_{\Box\!\Box}\otimes W_{\Box\!\Box}$ written in the basis $\{\tilde{e}_i\otimes\tilde{e}_j\}_{i,j=1,2,3}$ ordered lexicographically (that is, $\tilde{e}_1\otimes \tilde{e}_1$, $\tilde{e}_1\otimes \tilde{e}_2$, $\tilde{e}_1\otimes \tilde{e}_3$, $\tilde{e}_2\otimes \tilde{e}_1$, ...): 
\begin{equation}\label{mat-Sn}\left(\begin{array}{ccccccccc} 
\frac{(u-2)(u-1)}{u(u+1)}\!\!\!&\cdot&\cdot&\cdot&\cdot&\cdot&\cdot&\cdot&\cdot \\
\cdot&\frac{u-1}{u+1}&\cdot&\!\!\!\frac{-2(u-1)}{u(u+1)}&\cdot&\cdot&\cdot&\cdot&\cdot\\
\cdot&\cdot&1&\cdot&\frac{-4}{u+1}&\cdot&\frac{2}{u(u+1)}&\cdot&\cdot \\
\cdot&\frac{-2(u-1)}{u(u+1)}\!\!\!&\cdot&\frac{u-1}{u+1}&\cdot&\cdot&\cdot&\cdot &\cdot\\
\cdot&\cdot&\frac{-1}{u+1}&\cdot&\frac{u^2-u+2}{u(u+1)}&\cdot&\frac{-1}{u+1}&\cdot &\cdot\\
\cdot&\cdot&\cdot&\cdot&\cdot&\frac{u-1}{u+1}&\cdot&\!\!\!\frac{-2(u-1)}{u(u+1)}&\cdot \\
\cdot&\cdot&\frac{2}{u(u+1)}&\cdot&\frac{-4}{u+1}&\cdot&1&\cdot&\cdot \\
\cdot&\cdot&\cdot&\cdot&\cdot&\frac{-2(u-1)}{u(u+1)}\!\!\!&\cdot&\frac{u-1}{u+1}& \cdot\\
\cdot&\cdot&\cdot&\cdot&\cdot&\cdot&\cdot&\cdot& \!\!\!\frac{(u-2)(u-1)}{u(u+1)}\\
\end{array}\right)\,,\end{equation}
where the points indicate the coefficients equal to $0$. \hfill$\triangle$
\end{example}

\subsection{``Non-standard'' evaluations of the fusion function of $S_n$}\label{subsec-rem}

In the framework of the fusion procedure for the Yang solution, we are interested in any set of complex parameters $\bc=(c_1,\dots,c_n)$ such that the fusion function $\Phi$, defined by
\begin{equation}\label{phi-sym2}
\Phi(u_1,\dots,u_n):=\prod_{1\leq i<j\leq n}^{\rightarrow}(1-\frac{(i,j)}{u_i-u_j})\ ,
\end{equation}
gives a non-invertible element of $\C S_n$ when evaluated at $\bc$. Indeed, such a set of parameters $\bc$ will lead, via (\ref{SW-Sn}), to a non-invertible endomorphism $F(\bc)$, and in turn to a proper subspace $W_{\bc}\subset V^{\otimes n}$ such that $W_{\bc}\otimes V^{\otimes n'}$ is invariant for the fused operators of the form $R_{\bc,\un{\bc}}(u)$.
With the help of Theorem \ref{thm-idem-Sn}, we obtained in the preceding subsection such a set of parameters for any standard Young tableau (actually, in these situations, the evaluations of the fusion function are proportional to primitive idempotents of $\C S_n$).

One may easily see that there are other sets of parameters $\bc=(c_1,\dots,c_n)$ such that the evaluation of the fusion function gives an idempotent of $\C S_n$ (or more generally a non-invertible element of $\C S_n$). We will give in this subsection some examples for small $n$. We refer to \cite{Ch2,NT,Na4} for more information on evaluations of the fusion function associated to standard Young tableaux of skew shapes, and their relation to representations of Yangians.

\paragraph{\textbf{Case n=2.}} In this situation, the fusion function, with values in $\C S_2$, is given by
$$ \Phi(u_1,u_2)= 1 -\frac{(1,2)}{u_1-u_2}\ .$$
It is immediate to see that the only pairs of complex numbers $(c_1,c_2)$ such that the evaluation $\Phi(c_1,c_2)$ is non-singular and is a non-invertible element of $\C S_2$ are the pairs such that $c_2=c_1\pm1$ (this can be seen from (\ref{inv-R-Sn})). As the rational function $ \Phi(u_1,u_2)$ only depends on $u_1-u_2$, there is no loss of generality to consider that $c_1=0$. As a conclusion, for $n=2$, the only pairs $(c_1,c_2)$ leading to a proper subspace $W_{(c_1,c_2)}\subset V^{\otimes 2}$ are the pairs associated to the standard Young tableaux $\fbox{\scriptsize{1}}\fbox{\scriptsize{2}}$ and $\!\!\begin{array}{l}\fbox{\scriptsize{1}}\\[-0.2em] \fbox{\scriptsize{2}}\end{array}$.

\paragraph{\textbf{Case n=3.}} The fusion function with values in $\C S_3$ is given by
\begin{equation}\label{fus-S3}\Phi(u_1,u_2,u_3) =(1 -\frac{(1,2)}{u_1-u_2})(1 -\frac{(1,3)}{u_1-u_3})(1 -\frac{(2,3)}{u_2-u_3})\ .\end{equation}
As above, we can fix the freedom of translating all variables $u_1,u_2,u_3$ by the same constant with the assumption that $u_1$ is evaluated to $0$. 
\begin{proposition}\label{prop-rem}
The evaluation of the function $\Phi(0,u_2,u_3)$ at $u_2=c_2$ and $u_3=c_3$ is proportional to an idempotent of $\C S_3$ if and only if:
$$(c_2,c_3)\in\bigl\{(1,2),\,(1,-1),\,(-1,1),\,(-1,-2),\,(2,1),\,(-2,-1)\bigr\}\ .$$
\end{proposition}
\begin{proof} To prove the Proposition, we solve by a direct analysis the system of 6 equations in $x,c_2,c_3$, obtained by equating to 0 the coefficient in front of each element of $S_3$ in the expression $\bigl(\Phi(0,c_2,c_3)\bigr)^2-x\Phi(0,c_2,c_3)$ (we skip the details).
\end{proof}
The first four evaluations found in the preceding proposition correspond to standard tableaux, namely, with the same notation as in (\ref{ET}),
$$E_{\fbox{\tiny{1}}\fbox{\tiny{2}}\fbox{\tiny{3}}}\!=\!\frac{1}{6}\Phi(0,1,2)\,,\ \ \,
E_{\!\!\!\begin{array}{l}\fbox{\tiny{1}}\fbox{\tiny{2}}\\[-0.3em] \fbox{\tiny{3}}\end{array}}\!\!\!=\!\frac{1}{3}\Phi(0,1,-1)\,,\ \ \,
E_{\!\!\!\begin{array}{l}\fbox{\tiny{1}}\fbox{\tiny{3}}\\[-0.3em] \fbox{\tiny{2}}\end{array}}\!\!\!=\!\frac{1}{3}\Phi(0,-1,1)\,,\ \ \,
E_{\!\!\!\begin{array}{l}\fbox{\tiny{1}}\\[-0.3em] \fbox{\tiny{2}}\\[-0.3em] \fbox{\tiny{3}}\end{array}}\!\!\!=\!\frac{1}{6}\Phi(0,-1,-2)\,.$$
For brevity, let $E_1:=E_{\fbox{\tiny{1}}\fbox{\tiny{2}}\fbox{\tiny{3}}}$ and $E_4:=E_{\!\!\!\begin{array}{l}\fbox{\tiny{1}}\\[-0.3em] \fbox{\tiny{2}}\\[-0.3em] \fbox{\tiny{3}}\end{array}}\!\!$. For the two ``non-standard'' evaluations of the fusion function obtained in the above proposition, one can verify that the following elements
\begin{equation}\label{ns-idem}E_2:=\frac{1}{3}\Phi(0,2,1)\ \ \ \text{and}\ \ \ E_3:=\frac{1}{3}\Phi(0,-2,-1)\end{equation}
are idempotents of $\C S_3$ such that
$$E_iE_j=E_jE_i=\delta_{i,j} E_i\ \ \ \text{for $i,j=1,2,3,4$\ ,}\ \ \ \ \text{and}\ \ \ \ E_1+E_2+E_3+E_4=1_{S_3}\ ,$$
where $1_{S_3}$ is the unit element of $\C S_3$.  Thus, $\{E_1,E_2,E_3,E_4\}$ is a complete system of pairwise orthogonal primitive idempotents of $\C S_3$, and $E_2$ (respectively, $E_3$) generates a minimal left ideal isomorphic, as an $\C S_3$-module, to the irreducible module corresponding to the partition $\lambda=(2,1)$.

By the same arguments as in the preceding subsection, these two non-standard evaluations (namely, with $(c_2,c_3)=(2,1)$ or $(-2,-1)$) lead to fused solutions of the Yang--Baxter equation acting on spaces $M_{(2,1)}^{U(\mathfrak{gl}_N)}\otimes M_{\lambda'}^{U(\mathfrak{gl}_N)}$, where $\lambda'$ is any partition, alternative to the ones obtained in Corollary \ref{coro-fin-Sn}.

\begin{remark}
The two non-standard evaluations found in Proposition \ref{prop-rem} correspond actually to classical contents of standard Young tableaux, associated to skew partitions instead of usual partitions. Namely,  the non-standard evaluations correspond to $\!\!\!\begin{array}{l}\phantom{\fbox{\scriptsize{3}}}\fbox{\scriptsize{2}}\\[-0.2em] \fbox{\scriptsize{1}}\fbox{\scriptsize{3}}\end{array}$ and $\!\!\!\begin{array}{l}\phantom{\fbox{\scriptsize{3}}}\fbox{\scriptsize{1}}\\[-0.2em] \fbox{\scriptsize{2}}\fbox{\scriptsize{3}}\end{array}$. 
\hfill$\triangle$
\end{remark}

\begin{remark}
As we have already explained above, not only the evaluations of the fusion function providing elements proportional to idempotents are of interest for the fusion procedure; the evaluations of the fusion function giving non-invertible elements of $\C S_n$ furnish proper invariant subspaces as well. For $n=3$, with the notation (\ref{fus-S3}) and $u_1$ evaluated to 0, obvious examples of such evaluations (which include the ones studied previously) are 
$$\{u_2=\pm1,\,u_3\neq 0,u_2\}\,,\ \ \ \{u_3=\pm1,\,u_2\neq 0,u_3\}\,,\ \ \ \{u_3=u_2\pm1,\,u_2,u_3\neq0\}\ .$$
Further, one can verify that the two following consecutive evaluations
$$\Phi(0,u_2,u_3)\bigr\vert_{u_2=\pm1}\bigr\vert_{u_3=0}$$
also give non-invertible elements of $\C S_3$ (we also mention that the evaluations described in this remark exhaust the evaluations of $\Phi(0,u_2,u_3)$ providing non-invertible elements of $\C S_3$).\hfill$\triangle$
\end{remark}

\section{Invariant subspaces and representations of $U(\mathfrak{gl}_{N|M})$}\label{sec-Sn-s}

We explain how, with the help of a ``super" analogue of the Schur--Weyl duality \cite{BeRe,Se}, the fusion formula for the symmetric group can also be used in the context of the fusion procedure for a generalization of the Yang solution. We obtain a family of fused solutions acting on irreducible representations of the general linear Lie superalgebra $\mathfrak{gl}_{N|M}$.

\subsection{Sign conventions for $\Z/2\Z$-graded vector spaces and algebras}\label{subsec-sign}

Let $V=V_{\ov{0}}\oplus V_{\ov{1}}$ be a $\Z / 2\Z$-graded vector space. A non-zero vector $x\in V$ is called homogeneous if $x\in V_{\ov{0}}$ or $x\in V_{\ov{1}}\,$. For a homogeneous vector $x\in V_{\ov{i}}$, $i=0,1$, we set $|x|:=\ov{i}$. As soon as the notation $|x|$ appears, it is understood that $x$ is a homogeneous element. We will also use the notation $(-1)^{\ov{i}}$ with the obvious meaning: $(-1)^{\ov{0}}=1$ and $(-1)^{\ov{1}}=-1$.

If $W$ is another $\Z / 2\Z$-graded vector space then the tensor product $V\otimes W$ is considered as $\Z / 2\Z$-graded as well with the grading given by $|x\otimes y|=|x|+|y|$.

Let $A=A_{\ov{0}}\oplus A_{\ov{1}}$ be a superalgebra (\emph{i.e.} an algebra whose underlying vector space is $\Z/2\Z$-graded and such that the multiplication satisfies $A_{\ov{i}}A_{\ov{j}}\subset A_{\ov{i}+\ov{j}}$). A $\Z/ 2\Z$-graded vector space $V$ is an $A$-module if $V$ is an $A$-module for the ungraded algebra structure on $A$ satisfying moreover $A_{\ov{i}}(V_{\ov{j}})\subset V_{\ov{i}+\ov{j}}$.

If $A$ and $B$ are two superalgebras, the $\Z/2\Z$-graded vector space $A\otimes B$ is a superalgebra with the product defined by $(a\otimes b)(a'\otimes b')=(-1)^{|b||a'|}aa'\otimes bb'$.

Finally, if $V$ is an $A$-module and $W$ is a $B$-module then $V\otimes W$ has a structure of an $A\otimes B$-module for the action given by $(a\otimes b)(v\otimes w)=(-1)^{|b||v|}a(v)\otimes b(w)$.

\subsection{Generalization of the Yang solution for a $\Z / 2\Z$-graded vector space.}
We fix a $\Z / 2\Z$-decomposition of the vector space $V$ as $V=V_{\ov{0}}\oplus V_{\ov{1}}$ and we set $N:=\dim(V_{\ov{0}})$ and $M:=\dim(V_{\ov{1}})\,$. We define an operator $\tP\in\End(V\otimes V)$ by:
\begin{equation}\label{def-tP}
\tP(x\otimes y):= (-1)^{|x||y|} x\otimes y\ ,
\end{equation}
and we let $R$ be the following function of $u\in\C$ with values in $\End(V\otimes V)$:
\begin{equation}\label{Yang-s}
R(u):=\tP\Pe-\frac{\Pe}{u}\ ,
\end{equation}
where $\Pe$ is the usual permutation operator on $V\otimes V$. It is straightforward to check that the function $R$ is a solution of the Yang--Baxter equation on $V\otimes V$:
\begin{equation}\label{YB-add-s}
R_{1,2}(u)R_{1,3}(u+v)R_{2,3}(v)=R_{2,3}(v)R_{1,3}(u+v)R_{1,2}(u)\ .
\end{equation}
Moreover, this solution satisfies the same unitarity condition (\ref{inv-R-Sn}) as the Yang solution, namely
\begin{equation}\label{inv-R-Sn-s}
R(u)R_{2,1}(-u)=\frac{u^2-1}{u^2}\cdot\textrm{Id}_{V^{\otimes 2}}\,. 
\end{equation}

\subsection{Schur--Weyl duality in the $\Z / 2\Z$-graded setting.}
Let $n>1$. We recall that the symmetric group $S_n$ on $n$ letters is isomorphic to the group generated by elements $s_1,\dots,s_{n-1}$ subject to the defining relations:
\begin{equation}\label{Sn}
\begin{array}{lll}
s_i^2=1 && \text{for $i=1,\dots,n-1$\,,}\\[0.3em]
s_is_{i+1}s_i=s_{i+1}s_is_{i+1} && \text{for $i=1,\dots,n-2$\,,}\\[0.3em]
s_is_j=s_js_i && \text{for $i,j=1,\dots,n-1$ such that $|i-j|>1$\,,}
\end{array}
\end{equation}
the isomorphism being given by $s_i\mapsto (i,i+1)$, $i=1,\dots,n-1$.

The following map from the set of generators $\{s_1,\dots,s_{n-1}\}$ to the set $\End(V^{\otimes n})$,
\begin{equation}\label{rho-Sn-s}s_i\mapsto \tP_{i,i+1}\ \ \ \ \ \text{for $i=1,\dots,n-1$,}\end{equation}
extends to an algebra homomorphism from $\C S_n$ to $\End(V^{\otimes n})$ (one easily checks that the relations (\ref{Sn}) are satisfied by the images of the generators). We denote by $\tilde{\rho}$ this representation of $\C S_n$.

Let $\mathfrak{gl}_{N|M}$ be the Lie superalgebra of endomorphisms of $V$ and let  $U(\mathfrak{gl}_{N|M})$ be its universal enveloping algebra. We recall that $U(\mathfrak{gl}_{N|M})$ is a superalgebra and that its $\Z/2\Z$-grading is defined from the following $\Z/2\Z$-grading on $\mathfrak{gl}_{N|M}$: a non-zero element $g\in\mathfrak{gl}_{N|M}$ is homogeneous of degree $\ov{j}$ if $g(V_{\ov{i}})\subset V_{\ov{i}+\ov{j}}\,$, $i=0,1$.

The map, with values in $\End(V^{\otimes n})$, defined on elements $g\in\mathfrak{gl}_{N|M}$ by
$$g\mapsto g\otimes \textrm{Id} \otimes\dots\otimes \textrm{Id}+\textrm{Id} \otimes g\otimes \textrm{Id}\dots\otimes \textrm{Id}+\cdots\dots+\textrm{Id} \otimes\dots\otimes \textrm{Id}\otimes g\ .$$ 
extends to a superalgebras homomorphism from $U(\mathfrak{gl}_{N|M})$ to $U(\mathfrak{gl}_{N|M})^{\otimes n}$, which provides then, by composition with the natural representation of $U(\mathfrak{gl}_{N|M})^{\otimes n}$ on $V^{\otimes n}$, a representation of the superalgebra $U(\mathfrak{gl}_{N|M})$ on $V^{\otimes n}$. We denote by $\tilde{\chi}$ this representation (we recall that we use the sign convention of Section \ref{subsec-sign} for the action of $U(\mathfrak{gl}_{N|M})^{\otimes n}$ on $V^{\otimes n}$).

We recall the generalization of the Schur--Weyl duality, which holds between the symmetric group $S_n$ and the algebra $U(\mathfrak{gl}_{N|M})$  \cite{BeRe, Se}.
\begin{theorem}\label{thm-SW-Sn-s}
\begin{enumerate}
\item The subalgebra $\tilde{\rho}\bigl(\C S_n\bigr)$ of $\text{End}(V^{\otimes n})$ is the centraliser of $\tilde{\chi}\bigl(U(\mathfrak{gl}_{N|M}))$.
\item The subalgebra $\tilde{\chi}\bigl(U(\mathfrak{gl}_{N|M}))$ of $\text{End}(V^{\otimes n})$ is the centraliser of $\tilde{\rho}\bigl(\C S_n\bigr)$.
\item As an $\left(U(\mathfrak{gl}_{N|M})\!\otimes\!\C S_n\right)$-module defined by $\tilde{\chi}$ and $\tilde{\rho}$, the space $V^{\otimes n}$ decomposes as:
\begin{equation}\label{SW-Sn-for-s}
V^{\otimes n}\cong \bigoplus_{\lambda} M_{\lambda}^{U(\mathfrak{gl}_{N|M})}\otimes M_{\lambda}^{\C S_n}\ ,
\end{equation}
where $\lambda=(\lambda_1,\dots,\lambda_l)$ runs over the set of partitions such that $|\lambda|=n$ and $\lambda_j\leq M$ if $j>N$, and where the modules $M_{\lambda}^{U(\mathfrak{gl}_{N|M})}$ are irreducible $U(\mathfrak{gl}_{N|M})$-modules.
\end{enumerate}
\end{theorem}

\paragraph{\textbf{Fused solutions on $M_{\lambda}^{U(\mathfrak{gl}_{N|M})}\otimes M_{\lambda'}^{U(\mathfrak{gl}_{N|M})}$.}}

As a consequence of Formulas (\ref{Yang-s}) and (\ref{rho-Sn-s}), we have:
$$\tilde{\rho}\bigl(s_i -\frac{1}{u}\bigr)=R_{i,i+1}(u)\Pe=\hR_{i,i+1}(u)\ \ \ \ \ \text{for $i=1,\dots,n-1$\,.} $$
Therefore, defining the following rational function in the complex variables $u_1,\dots,u_n$ with values in $\C S_n$:
$$\tilde{\Phi}(u_1,\dots,u_n):=\prod_{i=1,\dots,n-1}^{\rightarrow}(s_i-\frac{1}{c_1-c_{i+1}})\dots(s_2-\frac{1}{c_{i-1}-c_{i+1}})(s_1-\frac{1}{c_i-c_{i+1}})\ ,$$
we have, recalling the definition (\ref{hF}),
\begin{equation}\label{trho-hF}\tilde{\rho}\bigl(\tilde{\Phi}(u_1,\dots,u_n)\bigr)=\hF(u_1,\dots,u_n)\ .\end{equation}
Moreover, the following Lemma relates the function $\tilde{\Phi}$ with the fusion function $\Phi$, defined by (\ref{phi-sym}), of the symmetric group.
\begin{lemma}\label{phi-tphi}
We have (where $w_n$ is the longest element of $S_n$):
$$\tilde{\Phi}(u_1,\dots,u_n)=\Phi(u_1,\dots,u_n)\,w_n\ .$$
\end{lemma}
\begin{proof}
Recall that the relation (\ref{ijk}) holds for the functions $(1-\frac{(i,j)}{u})$ with values in $\C S_n$. As moreover $\pi\cdot(i,j)=(\pi(i),\pi(j))\cdot\pi$ for $\pi\in S_n$, the lemma is a consequence of Lemma \ref{lem-F-alt} together with the fact that $w_n^2=1$ (one applies Lemma \ref{lem-F-alt} with $R_{i,j}(u)$ replaced by $(1-\frac{(i,j)}{u})$ and $\Pe_{\pi}$ replaced by $\pi$, for $\pi\in S_n$).
\end{proof}
Thus, as $w_n$ is invertible, the fusion formula for the symmetric group in Theorem \ref{thm-idem-Sn} can be used, as in Subsection \ref{subsec-conc-Sn}, to analyze the image of $\tilde{\rho}\bigl(\tilde{\Phi}(u_1,\dots,u_n)\bigr)$ in $V^{\otimes n}$ when $u_1,\dots,u_n$ are evaluated to classical contents of standard Young tableaux as in (\ref{ET}). 
 
Using Corollary \ref{cor-hF}, Formula (\ref{trho-hF}) and the Schur--Weyl duality stated in Theorem \ref{thm-SW-Sn-s}, we reproduce the same reasoning as in Subsection \ref{subsec-conc-Sn}, to obtain the generalization of the Corollary \ref{coro-fin-Sn} for the solution (\ref{Yang-s}). Namely, let $\mathcal{T}$, $\mathcal{T}'$ be two standard Young tableaux such that $\sh_{\cT}=(\lambda_1,\dots,\lambda_l)$ and $\sh_{\cT'}=(\lambda'_1,\dots,\lambda'_l)$, with $|\sh_{\cT}|=n$, $|\sh_{\cT'}|=n'$ and $\lambda_j,\lambda'_j\leq M$ if $j>N$. We obtain, by restriction of the fused operators to their invariant subspaces, a function $R^{\text{res}}_{\mathcal{T},\mathcal{T}'}$ with the following properties. 

\begin{corollary}\label{coro-fin-Sn-s}
\begin{enumerate}
\item The set of functions $\{R^{\text{res}}_{\cT,\cT'}\}$, where $\mathcal{T}$, $\mathcal{T}'$ are standard Young tableaux as above, forms a family of solutions of the Yang--Baxter equation, and $R^{\text{res}}_{\cT,\cT'}(u)$ is an endomorphism of a space isomorphic to $M_{\sh_{\cT}}^{U(\mathfrak{gl}_{N|M})}\otimes M_{\sh_{\cT'}}^{U(\mathfrak{gl}_{N|M})}$.
\item For four standard Young tableaux $\mathcal{T}$, $\widetilde{\mathcal{T}}$,  $\mathcal{T}'$ and $\widetilde{\mathcal{T}}'$ as above such that $\sh_{\phantom{\widetilde{T}}\!\!\!\!\cT}\!=\sh_{\widetilde{\cT}}$ and $\sh_{\phantom{\widetilde{T}}\!\!\!\!\cT'}\!=\sh_{\widetilde{\cT}'}$, the endomorphisms $R^{\text{res}}_{\cT,\cT'}(u)$ and $R^{\text{res}}_{\widetilde{\cT},\widetilde{\cT}'}(u)$ coincide up to a change of basis.
\end{enumerate}
\end{corollary}

\begin{remark}
Assume that $M>0$. Let $n=2$ and define 
$$\tilde{R}(u):=\tilde{\rho}(1-\frac{(1,2)}{u})=\textrm{Id}_{V^{\otimes 2}}-\frac{\tP}{u}\in\End(V^{\otimes 2})\ .$$
Note that, contrary to the situation $M=0$, the function $\tilde{R}$ is not a solution of the Yang--Baxter equation. This does not contradict the relation (\ref{ijk}) because if $n>2$ and $|i-j|>1$, we have $\tilde{R}_{i,j}(u)\neq \tilde{\rho}(1-\frac{(i,j)}{u})$. For example, we have $\tilde{\rho}(1-\frac{(1,3)}{u})=\tP_{2,3}\tilde{R}_{1,2}(u)\tP_{2,3}$ which is different from $\tilde{R}_{1,3}(u)=\Pe_{2,3}\tilde{R}_{1,2}(u)\Pe_{2,3}$.

It follows from (\ref{rho-Sn-s}) and (\ref{ijk}) that the function $\tilde{R}(u)$ satisfies, instead of (\ref{YB-add-s}), a ``braided" Yang--Baxter equation:
\begin{equation}\label{br-YB}
\tilde{R}_{1,2}(u)\tilde{R}_{1,\tilde{3}}(u+v)\tilde{R}_{\tilde{2},\tilde{3}}(v)=\tilde{R}_{\tilde{2},\tilde{3}}(v)\tilde{R}_{1,\tilde{3}}(u+v)\tilde{R}_{1,2}(u)\ ,
\end{equation}
with the braiding defined by the operator $\tP$, namely with $\tilde{R}_{1,\tilde{3}}(u):=\tP_{2,3}\tilde{R}_{1,2}(u)\tP_{2,3}$ and $\tilde{R}_{\tilde{2},\tilde{3}}(u):=\tP_{1,2}\tilde{R}_{1,\tilde{3}}(u)\tP_{1,2}$ (note that in our particular situation, we have $\tilde{R}_{\tilde{2},\tilde{3}}(u)=\tilde{R}_{2,3}(u)$).

Alternatively, setting $\hR(u)=R(u)\Pe=\tP-\displaystyle\frac{\textrm{Id}_{V^{\otimes 2}}}{u}$, the Yang--Baxter equation (\ref{YB-add-s}) is equivalent to 
\begin{equation}\label{YB-hat}
\hR_{1,2}(u)\hR_{2,3}(u+v)\hR_{1,2}(v)=\hR_{2,3}(v)\hR_{1,2}(u+v)\hR_{2,3}(u)\ .
\end{equation}
Then one can think as this one as the fundamental equation and, defining $\tilde{R}(u)=\hR(u)\tP$, derive (\ref{br-YB}) from (\ref{YB-hat}) (using again that here $\tilde{R}_{\tilde{2},\tilde{3}}(u)=\tP_{1,2}\tP_{2,3}\tilde{R}_{1,2}(u)\tP_{2,3}\tP_{1,2}=\tilde{R}_{2,3}(u)$).
\hfill$\triangle$
\end{remark}

\section{Invariant subspaces and representations of $U_q(\mathfrak{gl}_N)$}\label{sec-Hn}

In this last Section, we consider another class of examples of solutions of the Yang--Baxter equation. They are standard deformations of the solutions considered in Sections \ref{sec-Sn}-\ref{sec-Sn-s}. We present the generalization of the construction in Sections \ref{sec-Sn}-\ref{sec-Sn-s}, which leads here to a family of fused solutions of the Yang--Baxter equation acting on irreducible representations of $U_q(\mathfrak{gl}_{N|M})$ (the standard deformation of $U(\mathfrak{gl}_{N|M})$). We treat directly the general situation of a $\Z/2\Z$-graded vector space (the ungraded situation $M=0$, corresponding to the quantum group $U_q(\mathfrak{gl}_N)$, is covered as a particular case).

\subsection{Deformation of the Yang solution and its $\Z/2\Z$-graded analogue}\label{subsec-def-Yang}

As in Section \ref{sec-Sn-s}, we fix a $\Z / 2\Z$-decomposition of the vector space $V$ as $V=V_{\ov{0}}\oplus V_{\ov{1}}$ and we set $N:=\dim(V_{\ov{0}})$ and $M:=\dim(V_{\ov{1}})\,$. We use the notation introduced in Section \ref{subsec-sign}.

Let $(e_i)_{i=1,\dots,N+M}$ be a basis of the vector space $V$, such that $(e_1,\dots,e_N)$ is a basis of the subspace $V_{\ov{0}}$ and $(e_{N+1},\dots,e_{N+M})$ is a basis of the subspace $V_{\ov{1}}$.

Let $q$ be a non-zero complex number and let $\hR\in\End(V\otimes V)$ be defined by, for $i,j=1,\dots,N+M$,
\begin{equation}\label{hR-Hn-s}
\hR(e_i\otimes e_j):=\left\{\begin{array}{ll}
\displaystyle\frac{(-1)^{|e_i|}(q+q^{-1})+(q-q^{-1})}{2}\,e_i\otimes e_j\ \  & \text{if $i=j$,}\\[0.8em]
(-1)^{|e_i||e_j|}\,e_j\otimes e_i+(q-q^{-1})\,e_i\otimes e_j & \text{if $i<j$,}\\[0.4em]
(-1)^{|e_i||e_j|}\,e_j\otimes e_i & \text{if $i>j$.}
\end{array}\right.
\end{equation}

It is known \cite{Mi,Moon}, and it can be directly checked, that $\hR$ satisfies the quadratic relation
\begin{equation}\label{quadr-R}
\hR^2-(q-q^{-1})\hR-\textrm{Id}_{V^{\otimes 2}}=0\ ,
\end{equation}
and verifies as well, on the space $V\otimes V\otimes V$,
\begin{equation}\label{braid-R}
\hR_{1,2}\hR_{2,3}\hR_{1,2}=\hR_{2,3}\hR_{1,2}\hR_{2,3}\ .
\end{equation}
Relations (\ref{quadr-R}) and (\ref{braid-R}) imply by a direct calculation that the function of $\alpha\in\C$, with values in $\End(V\otimes V)$, given by:
\begin{equation}\label{hRa-Hn}
\hR(\alpha):=\hR+(q-q^{-1})\frac{\textrm{Id}_{V^{\otimes 2}}}{\alpha^{-1}-1}\ ,
\end{equation}
satisfies the equation (operators act on $V\otimes V\otimes V$)
$$\hR_{1,2}(\alpha)\hR_{2,3}(\alpha \beta)\hR_{1,2}(\beta)=\hR_{2,3}(\beta)\hR_{1,2}(\alpha \beta)\hR_{2,3}(\alpha)\ .$$
As a direct consequence, the function $R$ given by:
\begin{equation}\label{Ra-Hn}
R(\alpha):=\hR(\alpha)\Pe\ ,
\end{equation}
is a solution, on $V\otimes V$, of the Yang--Baxter equation with multiplicative spectral parameters:
$$R_{1,2}(\alpha)R_{1,3}(\alpha \beta)R_{2,3}(\beta)=R_{2,3}(\beta)R_{1,3}(\alpha \beta)R_{1,2}(\alpha)\ .$$

\begin{remark}\label{rem-def}
The solution $R(\alpha)$ is a deformation of the solution (\ref{Yang-s}) (and in particular of the Yang solution (\ref{Yang2}) if $M=0$) in the following sense. Consider $q$ as a variable in $\C\backslash\{0\}$ and set $\alpha=q^{2u}$. Since $\hR\bigr\vert_{q=1}=\tP$, we obtain
$$\hR(q^{2u})\bigr\vert_{q=1}=\tP-\frac{\textrm{Id}_{V^{\otimes 2}}}{u}\ \ \ \ \ \text{and}\ \ \ \ \ R(q^{2u})\bigr\vert_{q=1}=\tP\Pe-\frac{\Pe}{u}\ .$$
\vskip -0.6cm\hfill$\triangle$
\end{remark}

\begin{example}\label{ex-R-Hn}
In this example, let $N+M=2$. We write the matrix of the endomorphism $R(\alpha)$ in the basis $\{e_1\otimes e_1,\,e_1\otimes e_2,\,e_2\otimes e_1,\,e_2\otimes e_2\}$ of $V\otimes V$ (points indicate coefficients equal to $0$):
\[\left(\begin{array}{cccc}
\displaystyle\frac{q\alpha^{-1}-q^{-1}}{\alpha^{-1}-1} & \cdot & \cdot & \cdot \\
\cdot  & 1 & \displaystyle\frac{(q-q^{-1})\alpha^{-1}}{\alpha^{-1}-1}& \cdot \\
\cdot  & \displaystyle\frac{q-q^{-1}}{\alpha^{-1}-1} & 1& \cdot \\
\cdot & \cdot & \cdot & x
\end{array}\right)\ \ \ \text{with}\ \ x=\left\{\begin{array}{cc}
\displaystyle\frac{q\alpha^{-1}-q^{-1}}{\alpha^{-1}-1} & \text{if $N|M=2|0$\,,}\\[1em]
\displaystyle\frac{q-\alpha^{-1}q^{-1}}{\alpha^{-1}-1} & \text{if $N|M=1|1$\,.}
\end{array}\right.\]
\vskip -0.5cm\hfill$\triangle$
\end{example}

\subsection{Jimbo--Schur--Weyl duality and its $\Z/2\Z$-graded analogue}

From now on, we assume that $q\in\C\backslash\{0\}$ is not a root of unity. We will use in all the following the conventions recalled in Section \ref{subsec-sign} for the superalgebras $U_q(\mathfrak{gl}_{N|M})$, its tensor powers and its representations.

\paragraph{\textbf{Hecke algebra.}}
The Hecke algebra (of type A)  is the associative algebra $H_n(q)$ over $\C$ generated by $\sigma_1,\dots,\sigma_{n-1}$ with the defining relations:
\begin{equation}\label{Hecke}
\begin{array}{lll}
\sigma_i^2=(q-q^{-1})\sigma_i+1 && \text{for $i=1,\dots,n-1$\,,}\\[0.3em]
\sigma_i\sigma_{i+1}\sigma_i=\sigma_{i+1}\sigma_i\sigma_{i+1} && \text{for $i=1,\dots,n-2$\,,}\\[0.3em]
\sigma_i\sigma_j=\sigma_j\sigma_i && \text{for $i,j=1,\dots,n-1$ such that $|i-j|>1$\,.}
\end{array}
\end{equation}
The following map from the set of generators of $H_n(q)$ to $\End(V^{\otimes n})$:
\begin{equation}\label{SW-Hn}
\sigma_i \mapsto \hR_{i,i+1}\ \quad\text{for $i=1,\dots,n-1$\,,}
\end{equation}
extends to an algebra homomorphism. This follows from Relations (\ref{quadr-R}) and (\ref{braid-R}), together with the obvious commutation relation $\hR_{i,i+1}\hR_{j,j+1}=\hR_{j,j+1}\hR_{i,i+1}$ if $|i-j|>1$. We denote by $\rho$ this representation of $H_n(q)$ on the space $V^{\otimes n}$.

\paragraph{\textbf{Quantum algebra $U_q(\mathfrak{gl}_{N|M})$.}}
The standard deformation of $U(\mathfrak{gl}_{N|M})$ is the superalgebra $U_q(\mathfrak{gl}_{N|M})$ defined by generators and relations as follows. The generators are 
$$K_i^{\pm1}\,,\ i=1,\dots,N+M\,,\ \ \ \text{and}\ \ \ E_j,F_j\,,\ j=1,\dots,N+M-1\,,$$ 
with the grading defined by $|E_N|=|F_N|=\ov{1}$ and all other generators are of degree $\ov{0}$.

Let $a_{i,i}=-a_{i+1,i}=1$, $i=1,\dots,N-1$, and $a_{i,j}=0$ otherwise. The defining relations are:
$$ \begin{array}{lll}
K_iK_i^{-1}=K_i^{-1}K_i=1\,,\quad K_iK_j=K_jK_i\,, && \text{$i,j\in\{1,\dots,N+M\}$\,,}
\\[0.5em]
K_iE_jK_i^{-1}=q^{(-1)^{|e_i|}a_{i,j}}E_j\,, && \text{$i\in\{1,\dots,N+M\}$, $j\in\{1,\dots,N+M-1\}$\,,}\\[0.5em]
K_iF_jK_i^{-1}=q^{-(-1)^{|e_i|}a_{i,j}}F_j\,, && \text{$i\in\{1,\dots,N+M\}$, $j\in\{1,\dots,N+M-1\}$\,,}\\[0.5em]
E_iF_j-F_jE_i=0\,, && \text{$i,j\in\{1,\dots,N+M-1\}$ with $i\neq j$\,,}\\[0.5em]
E_iE_j=E_jE_i\,,\quad F_iF_j=F_jF_i\,, && \text{$i,j\in\{1,\dots,N+M-1\}$ with $|i-j|>1$\,,}
\end{array}$$
together with, for $i\in\{1,\dots,N+M-1\}$ such that $i\neq N$,
$$ \begin{array}{lll}
E_iF_i-F_iE_i=(-1)^{|e_i|}\displaystyle\frac{K_iK_{i+1}^{-1}-K_i^{-1}K_{i+1}}{q-q^{-1}}\,, &&\\[1em]
E_jE_i^2-(q+q^{-1})E_iE_jE_i+E_i^2E_{j}=0\,, && \text{$j\in\{1,\dots,N+M-1\}$ with $|i-j|=1$\,,}\\[0.5em]
F_jF_i^2-(q+q^{-1})F_iF_jF_i+F_i^2F_{j}=0\,, && \text{$j\in\{1,\dots,N+M-1\}$ with $|i-j|=1$\,;}
\end{array}$$
we note that these relations are enough if $M=0$ (for the quantum group $U_q(\mathfrak{gl}_N)$). The remaining relations are
$$ \begin{array}{l}
E_N^2=F_N^2=0\,, \\[0.5em]
E_NF_N+F_NE_N=\displaystyle\frac{K_NK_{N+1}^{-1}-K_N^{-1}K_{N+1}}{q-q^{-1}}\,,\\[1em]
E_NE_{N-1}E_NE_{N+1}+E_NE_{N+1}E_NE_{N-1}+E_{N-1}E_NE_{N+1}E_N+E_{N+1}E_NE_{N-1}E_N\\
=(q+q^{-1})E_NE_{N-1}E_{N+1}E_N\,,\\[0.5em]
F_NF_{N-1}F_NF_{N+1}+F_NF_{N+1}F_NF_{N-1}+F_{N-1}F_NF_{N+1}F_N+F_{N+1}F_NF_{N-1}F_N\\
=(q+q^{-1})F_NF_{N-1}F_{N+1}F_N\,,
\end{array}$$
the first two relations being present only if $N$ and $M$ are different from 0, and the last two relations being present only if $N>1$ and $M>1$.

The natural representation $\eta$ of the superalgebra $U_q(\mathfrak{gl}_{N|M})$ on $V$ is given, on the basis $\{e_k\}_{k=1,\dots,N+M}$, by:
$$\begin{array}{lcll}
\eta(K_i)\bigl(e_i\bigr)=q^{(-1)^{|e_i|}}e_i\  & \ \text{and}\  & \eta(K_i)\bigl(e_k\bigr)=e_k\ \ & \text{if $k\neq i$\,,}\\[0.5em]
\eta(E_j)\bigl(e_{j+1}\bigr)=e_j & \text{and} & \eta(E_j)\bigl(e_k\bigr)=0 & \text{if $k\neq j+1$\,,}\\[0.5em]
\eta(F_j)\bigl(e_j\bigr)=e_{j+1} &\text{and} & \eta(F_j)\bigl(e_k\bigr)=0 & \text{if $k\neq j$\,.}
\end{array}$$

The superalgebra $U_q(\mathfrak{gl}_{N|M})$ admits a coproduct, which is the superalgebras homomorphism $\Delta$ from $U_q(\mathfrak{gl}_{N|M})$ to $U_q(\mathfrak{gl}_{N|M})\otimes U_q(\mathfrak{gl}_{N|M})$ defined on the generators by:
$$\begin{array}{lll}
\Delta(K^{\pm1}_i)=K^{\pm1}_i\otimes K^{\pm1}_i &&\text{for $i=1,\dots,N+M$\,,}\\[0.5em]
\Delta(E_j)=E_j\otimes K_j^{-1}\!K_{j+1}+1\otimes E_j && \text{for $j=1,\dots,N+M-1$\,.}\\[0.5em]
\Delta(F_j)=F_j\otimes 1+K_jK_{j+1}^{-1}\otimes F_j && \text{for $j=1,\dots,N+M-1$\,.}
\end{array}$$

Via the coproduct $\Delta$, the representation $\eta$ induces a representation $\chi$ of the algebra $U_q(\mathfrak{gl}_{N|M})$ on the space $V^{\otimes n}$. More precisely, let $\Delta^{(2)}:=\Delta$ and define inductively $\Delta^{(k)}:=(\Delta^{(k-1)}\otimes \textbf{1})\circ\Delta$ for $k=3,\dots,n$, where $\textbf{1}$ is the identity homomorphism of  $U_q(\mathfrak{gl}_{N|M})$. Explicitly, we have:
$$\begin{array}{cll}
\Delta^{(k)}(K_i^{\pm1})=K_i^{\pm1}\otimes K_i^{\pm1}\otimes\dots\otimes K_i^{\pm1}&& \text{for $i=1,\dots,N+M$\,,}\\[0.5em]
\Delta^{(k)}(E_j)=\displaystyle\sum_{p=0,\dots,k-1}1^{\otimes p}\otimes E_j\otimes (K_j^{-1}\!K_{j+1})^{\otimes k-1-p} & & \text{for $j=1,\dots,N+M-1$\,,}\\[1.4em]
\Delta^{(k)}(F_j)=\displaystyle\sum_{p=0,\dots,k-1}(K_jK_{j+1}^{-1})^{\otimes p}\otimes F_j\otimes 1^{\otimes k-1-p} & & \text{for $j=1,\dots,N+M-1$\,.}
\end{array}$$
Then the representation $\chi$ of the algebra $U_q(\mathfrak{gl}_{N|M})$ on the space $V^{\otimes n}$ is given by:
$$ \chi(X):=(\eta\otimes\eta\otimes\dots\otimes\eta)\circ\Delta^{(n)}(X)\ \quad\text{for any $X\in U_q(\mathfrak{gl}_{N|M})$\,.}$$

\paragraph{\textbf{Jimbo--Schur--Weyl duality and its $\Z/2\Z$-graded analogue.}}
To state the analogue of the Schur--Weyl dualities (Theorems \ref{thm-SW-Sn} and \ref{thm-SW-Sn-s}) between the Hecke algebra $H_n(q)$ and the quantum algebra $U_q(\mathfrak{gl}_{N|M})$, we recall that the irreducible representations of $H_n(q)$ are parametrized by the partitions of $n$. For any partition $\lambda$ of $n$, we denote by $M_{\lambda}^{H_n(q)}$ the corresponding irreducible $H_n(q)$-module (with the convention that $M_{(n)}^{H_n(q)}$ is the one-dimensional $H_n(q)$-module on which the generators $\sigma_1,\dots,\sigma_{n-1}$ act by multiplication by $q$).

The analogue of the Schur--Weyl duality consists in the following assertions (see \cite{Ji1} for the case $M=0$, and see \cite{Mi,Moon,TK} for the general case).
\begin{theorem}\label{thm-SW-Hn}
\begin{enumerate}
\item The subalgebra $\rho\bigl(H_n(q)\bigr)$ of $\text{End}(V^{\otimes n})$ is the centraliser of $\chi\bigl(U_q(\mathfrak{gl}_{N|M})\bigr)$.
\item The subalgebra $\chi\bigl(U_q(\mathfrak{gl}_{N|M})\bigr)$ of $\text{End}(V^{\otimes n})$ is the centraliser of $\rho\bigl(H_n(q)\bigr)$.
\item As an $\left(U_q(\mathfrak{gl}_{N|M})\!\otimes\!H_n(q)\right)$-module defined by $\chi$ and $\rho$, the space $V^{\otimes n}$ decomposes as:
\begin{equation}\label{SW-Hn-for}
V^{\otimes n}\cong \bigoplus_{\lambda} M_{\lambda}^{U_q(\mathfrak{gl}_{N|M})}\otimes M_{\lambda}^{H_n(q)}\ ,
\end{equation}
where $\lambda=(\lambda_1,\dots,\lambda_l)$ runs over the set of partitions such that $|\lambda|=n$ and $\lambda_j\leq M$ if $j>N$, and where $M_{\lambda}^{U_q(\mathfrak{gl}_{N|M})}$ are irreducible $U_q(\mathfrak{gl}_{N|M})$-modules.
\end{enumerate}
\end{theorem}

\subsection{Fusion formula for idempotents of the Hecke algebra}

Let $n>1$ and recall the definition (\ref{hF}) of the function $\hF$ of $\bc:=(c_1,\dots,c_n)\in\C^n$ (with the multiplicative version for the spectral parameters, see Remark \ref{rem-mult}) with values in $\End(V^{\otimes n})$:
\begin{equation}\label{hF2}
\hF(\bc):=\prod_{i=1,\dots,n-1}^{\rightarrow}\hR_{i,i+1}(\frac{c_1}{c_{i+1}})\dots\hR_{2,3}(\frac{c_{i-1}}{c_{i+1}})\hR_{1,2}(\frac{c_i}{c_{i+1}})\ .
\end{equation}
Due to Corollary \ref{cor-hF}, if the function $\hF$ gives for a particular value $\bc:=(c_1,\dots,c_n)$ a well-defined endomorphism $\hF(\bc)$ of $V^{\otimes n}$, then the subspace $W_{\bc}\subset V^{\otimes n}$ is the image of $\hF(\bc)$.

Following Formulas (\ref{hF2}), (\ref{SW-Hn}) and (\ref{hRa-Hn}), we define the following rational function in variables $\alpha_1,\dots,\alpha_n$ with values in the algebra $H_n(q)$:
\begin{equation}\label{phi-Hn}
\Psi(\alpha_1,\dots,\alpha_n):=\prod_{i=1,\dots,n-1}^{\rightarrow}\sigma_i(\frac{\alpha_1}{\alpha_{i+1}})\dots\sigma_2(\frac{\alpha_{i-1}}{\alpha_{i+1}})\sigma_1(\frac{\alpha_i}{\alpha_{i+1}})\,,
\end{equation}
where $\sigma_i(\alpha):=\displaystyle\sigma_i+\frac{q-q^{-1}}{\alpha^{-1}-1}$ for all $i=1,\dots,n-1$. The elements $\sigma_i(\alpha)$ with values in $H_n(q)$ are the \emph{Baxterized elements} of the Hecke algebra.
We also set
$$T_{w_n}:= \sigma_1(\sigma_2\sigma_1)\dots\dots(\sigma_{n-2}\dots\sigma_2\sigma_1)(\sigma_{n-1}\dots\sigma_2\sigma_1)=\prod_{i=1,\dots,n-1}^{\rightarrow}\!\!\sigma_i\dots\sigma_2\sigma_1\ .$$

The main result concerning the function $\Psi(u_1,\dots,u_n)$ is the following (we follow \cite{IMOs}). Let $\lambda$ be a partition of $n$ and 
let $\mathcal{T}$ be a standard Young tableau of shape $\lambda$. For brevity, set $\qc_i:=\qc(\mathcal{T}|i)$ for $i=1,\dots,n$. We recall that $f^{(q)}(\lambda)$ is the non-zero complex number defined by (\ref{fq-lam}).
\begin{theorem}\label{thm-idem-Hn}
The element obtained by the following consecutive evaluations (from $\alpha_1$ to $\alpha_n$)
\begin{equation}\label{idem-Hn}
E^{(q)}_{\mathcal{T}}=
f^{(q)}(\lambda)\Psi(\alpha_1,\dots,\alpha_n)T_{w_n}^{-1}\Bigr\vert_{\alpha_1=\qc_1}\Bigr\vert_{\alpha_2=\qc_2}\dots \Bigr\vert_{\alpha_n=\qc_n}
\end{equation}
is a primitive idempotent of $H_n(q)$ which generates a minimal left ideal isomorphic, as an $H_n(q)$-module, to the irreducible module $M_{\lambda}^{H_n(q)}$.
\end{theorem}

\begin{remark}\label{rem-cl}
In the same spirit as in Remark \ref{rem-def}, we verify that
$$\sigma_i(q^{2u})\bigr\vert_{q=1}=(i,i+1)-\frac{1}{u}\ \ \ \ \ \text{for $i=1,\dots,n-1$,}$$
where we recall that $H_n(1)\cong \C S_n$ with the isomorphism given by $\sigma_i\bigr\vert_{q=1}\mapsto (i,i+1)$, $i=1,\dots,n-1$. Thus, setting $\alpha_j=q^{2u_j}$, for $j=1,\dots,n$, we obtain:
$$\Psi(\alpha_1,\dots,\alpha_n)\bigr\vert_{q=1}\!=\!\!\!\!\!\!\prod_{i=1,\dots,n-1}^{\rightarrow}\!\!\!\bigl((i,i+1)-\frac{1}{u_1-u_{i+1}}\bigr)\dots\bigl((2,3)-\frac{1}{u_{i-1}-u_{i+1}}\bigr)\bigl((1,2)-\frac{1}{u_i-u_{i+1}}\bigr)\,.$$
Using Lemma \ref{phi-tphi}, we conclude that
$$\Psi(\alpha_1,\dots,\alpha_n)T_{w_n}^{-1}\Bigr\vert_{q=1}=\Phi(u_1,\dots,u_n)\ ,$$
where $\Phi$ is the fusion function of the symmetric group defined by (\ref{phi-sym}). In this sense, Theorem \ref{thm-idem-Sn} is the classical limit ($q\to 1$) of Theorem \ref{thm-idem-Hn}.
\hfill $\triangle$
\end{remark}

\subsection{Fused solutions on $M_{\lambda}^{U_q(\mathfrak{gl}_{N|M})}\otimes M_{\lambda'}^{U_q(\mathfrak{gl}_{N|M})}$}\label{subsec-fin-Hn}

Let $\mathcal{T}$ be a standard Young tableau such that $\sh_{\cT}=\lambda=(\lambda_1,\dots,\lambda_l)$ with $|\lambda|=n$ and $\lambda_j\leq M$ if $j>N$. We let
\begin{equation}\label{def-hFT} \hF(\cT):=\hF(\bc)\Bigr\vert_{c_1=\qc(\mathcal{T}|1)}\Bigr\vert_{c_2=\qc(\mathcal{T}|2)}\dots \Bigr\vert_{c_n=\qc(\mathcal{T}|n)}\ ,\end{equation}
where $\bc=(c_1,\dots,c_n)$ is seen as an $n$-tuple of variables.

According to Formulas (\ref{hF2}), (\ref{phi-Hn}) and (\ref{idem-Hn}), we have that
$$\hF(\cT)=\rho\left(\bigl(f^{(q)}(\lambda)\bigr)^{-1}  E^{(q)}_{\mathcal{T}}T_{w_n}\right)\ .$$
The image $W_{\cT}\subset V^{\otimes n}$ of the operator $\hF(\cT)$ is well-defined and coincides, since $T_{w_n}$ is invertible, with the image in $V^{\otimes n}$ of the operator $\rho\left(E^{(q)}_{\mathcal{T}}\right)$. With a similar reasoning as in Subsection \ref{subsec-conc-Sn} before Theorem \ref{thm-conc-Sn}, we conclude that Theorem \ref{thm-idem-Hn} together with Formula (\ref{SW-Hn-for}) in Theorem \ref{thm-SW-Hn} implies:
\begin{theorem}\label{thm-conc-Hn}
The subspace $W_{\cT}$ of $V^{\otimes n}$ is an irreducible $U_q(\mathfrak{gl}_{N|M})$-module isomorphic to $M_{\lambda}^{U_q(\mathfrak{gl}_{N|M})}$.
\end{theorem}

Let $k\in\{1,\dots,n-1\}$ and denote by $\mathcal{T}^{(s_k)}$ the Young tableau of shape $\lambda$ obtained from $\mathcal{T}$ by exchanging the nodes with numbers $k$ and $k+1$. We assume that $\mathcal{T}^{(s_k)}$ is also standard, which is equivalent to the fact that $\qc(\cT|k+1)\neq \qc(\cT|k)q^{\pm2}$. The solution (\ref{Ra-Hn}) satisfies the following unitarity condition, implied by (\ref{quadr-R}),
\begin{equation}\label{inv-R-Hn}R(\alpha)R_{2,1}(\frac{1}{\alpha})=\frac{(\alpha-q^2)(\alpha-q^{-2})}{(\alpha-1)^2}\cdot\textrm{Id}_{V^{\otimes 2}}\,. \end{equation}
Therefore, according to Definition \ref{def-adm}, the condition for the tableau $\mathcal{T}^{(s_k)}$ to be standard is equivalent to the condition for $s_k=(k,k+1)$ to be an admissible transposition for $R(\alpha)$ and the set of parameters $\bc^{\cT}$. 

As in Subsection \ref{subsec-conc-Sn}, Corollary \ref{coro-fin-Sn}, we sum up the results obtained for the standard deformation of the Yang solution.
Let $\mathcal{T}$ and $\mathcal{T}'$ be two standard Young tableaux such that $\sh_{\cT}=\lambda=(\lambda_1,\dots,\lambda_l)$ and $\sh_{\cT'}=\lambda'=(\lambda'_1,\dots,\lambda'_{l'})$ satisfying $|\lambda|=n$, $|\lambda'|=n'$ and $\lambda_j,\lambda'_j\leq M$ if $j>N$. We set $\bc^{\cT}\!\!:=\bigl(\qc(\mathcal{T}|1),\dots,\qc(\mathcal{T}|n)\bigr)$, $\un{\bc}^{\cT'}\!\!:=\bigl(\qc(\mathcal{T'}|1),\dots,\qc(\mathcal{T'}|n')\bigr)$ and we define $R^{\text{res}}_{\mathcal{T},\mathcal{T}'}(u)$ to be the restriction of the operator $R_{\bc^{\cT},\un{\bc}^{\cT'}}(u)$ to the subspace $W_{\cT}\otimes W_{\cT'}\subset V^{\otimes n}\otimes V^{\otimes n'}$:
$$R^{\text{res}}_{\mathcal{T},\mathcal{T}'}(u):=R_{\bc^{\cT},\un{\bc}^{\cT'}}(u)\Bigr\vert_{W_{\cT}\otimes W_{\cT'}}\ .$$

\begin{corollary}\label{coro-fin-Hn}
\begin{enumerate}
\item The set of functions $\{R^{\text{res}}_{\cT,\cT'}\}$, where $\mathcal{T}$, $\mathcal{T}'$ are standard Young tableaux as above, forms a family of solutions of the Yang--Baxter equation, where $R^{\text{res}}_{\cT,\cT'}(\alpha)$ is an endomorphism of a space isomorphic to $M_{\sh_{\cT}}^{U_q(\mathfrak{gl}_{N|M})}\otimes M_{\sh_{\cT'}}^{U_q(\mathfrak{gl}_{N|M})}$.
\item For four standard Young tableaux $\mathcal{T}$, $\widetilde{\mathcal{T}}$,  $\mathcal{T}'$ and $\widetilde{\mathcal{T}}'$ as above such that $\sh_{\cT}=\sh_{\widetilde{\cT}}$ and  $\sh_{\cT'}=\sh_{\widetilde{\cT}'}$, the endomorphisms $R^{\text{res}}_{\cT,\cT'}(\alpha)$ and $R^{\text{res}}_{\widetilde{\cT},\widetilde{\cT}'}(\alpha)$ coincide up to a change of basis.
\end{enumerate}
\end{corollary}

\begin{example}\label{ex-fus-Hn}
This example is the (deformed) analogue of Example \ref{ex-fus-Sn}. Namely, we consider the fused operator $R_{\bc,\un{\bc}}(\alpha)$, with $n=n'=2$, $c_1=c_{\un{1}}=1$ and $c_2=c_{\un{2}}=q^2$, obtained from the solution (\ref{Ra-Hn}). The expression for $R_{\bc,\un{\bc}}(\alpha)$ is
$$R_{\bc,\un{\bc}}(\alpha)=R_{2,\un{1}}(\alpha q^2)R_{1,\un{1}}(\alpha)R_{2,\un{2}}(\alpha)R_{1,\un{2}}(\alpha q^{-2})\ .$$
The results of this section, applied to this example, give that the operator $R_{\bc,\un{\bc}}(\alpha)$ preserves the subspace $W^{(q)}_{\Box\!\Box}\otimes W_{\Box\!\Box}^{(q)}\subset V^{\otimes 2}\otimes V^{\otimes 2}$, where $W^{(q)}_{\Box\!\Box}$ is the space of the irreducible representation of $U_q(\mathfrak{gl}_N)$ corresponding to the partition $\lambda=(2)$. The subspace $W^{(q)}_{\Box\!\Box}\subset V^{\otimes 2}$ is the image of the operator $R(q^{-2})$.

Let $N|M=2|0$ and fix a basis $\{e_1,e_2\}$ of $V$. From the matrix in Example \ref{ex-R-Hn}, with $\alpha=q^{-2}$, we find that the following vectors form a basis of $W^{(q)}_{\Box\!\Box}$:
$$\tilde{e}_1:=e_1\otimes e_1\,,\ \ \ \  \tilde{e}_2:=e_1\otimes e_2+q\,e_2\otimes e_1\ \ \ \ \text{and}\ \ \ \ \tilde{e}_3:=e_2\otimes e_2\ .$$

We give the matrix of the endomorphism $R_{\bc,\un{\bc}}(\alpha)$ on the space $W^{(q)}_{\Box\!\Box}\otimes W^{(q)}_{\Box\!\Box}$ written in the basis $\{\tilde{e}_i\otimes\tilde{e}_j\}_{i,j=1,\dots,3}$ ordered lexicographically (the points indicate the coefficients equal to $0$): 
$$\left(\begin{array}{ccccccccc} 
\frac{\al{-2}\al{-1}}{\al{0}\al{1}}\!\!\!&\cdot&\cdot&\cdot&\cdot&\cdot&\cdot&\cdot&\cdot \\
\cdot&\frac{\al{-1}}{\al{1}}&\cdot&\frac{-[2]\al{-1}\alpha}{\al{0}\al{1}}&\cdot&\cdot&\cdot&\cdot&\cdot\\
\cdot&\cdot&1&\cdot&\frac{-q [2]^2\alpha}{\al{1}}&\cdot&\!\!\!\frac{[2]\alpha^2}{\al{0}\al{1}}\!\!\!&\cdot&\cdot \\
\cdot&\frac{-[2]\al{-1}}{\al{0}\al{1}}&\cdot&\frac{\al{-1}}{\al{1}}&\cdot&\cdot&\cdot&\cdot &\cdot\\
\cdot&\cdot&\frac{-q^{-1}}{\al{1}}&\cdot&\frac{P(\alpha)}{\al{0}\al{1}}&\cdot&\frac{-q^{-1}}{\al{1}}&\cdot &\cdot\\
\cdot&\cdot&\cdot&\cdot&\cdot&\frac{\al{-1}}{\al{1}}&\cdot&\frac{-[2]\al{-1}\alpha}{\al{0}\al{1}}\!\!\!&\cdot \\
\cdot&\cdot&\!\!\!\frac{[2]}{\al{0}\al{1}}\!\!\!&\cdot&\frac{-q [2]^2}{\al{1}}&\cdot&1&\cdot&\cdot \\
\cdot&\cdot&\cdot&\cdot&\cdot&\frac{-[2]\al{-1}}{\al{0}\al{1}}&\cdot&\frac{\al{-1}}{\al{1}}& \cdot\\
\cdot&\cdot&\cdot&\cdot&\cdot&\cdot&\cdot&\cdot& \frac{\al{-2}\al{-1}}{\al{0}\al{1}}\\
\end{array}\right)\,,$$
where we have set ($n\in\mathbb{Z}$)
$$[n]:=\displaystyle\frac{q^n-q^{-n}}{q-q^{-1}}\,,\ \ \ \al{n}:=\displaystyle\frac{q^n\alpha-q^{-n}}{q-q^{-1}}\ \ \ \ \text{and}\ \ \ \ P(\alpha):=\displaystyle\frac{q^{-1}\alpha^2+(q^3-2q-2q^{-1}+q^{-3})\alpha+q}{(q-q^{-1})^2}\,.$$
 Taking $\alpha=q^{2u}$ and performing the ``classical'' limit ($q\to 1$) on the coefficients of the above matrix, we obtain the matrix (\ref{mat-Sn}) displayed in Example \ref{ex-fus-Sn} (as it should be in view of Remark \ref{rem-def}). To perform the classical limit, we use that:
$$[n]\Bigr\vert_{q=1}=n\,,\ \ \ \ [n]_{q^{2u}}\Bigr\vert_{q=1}=u+n\ \ \ \ \text{and}\ \ \ \  P(q^{2u})\Bigr\vert_{q=1}=u^2-u+2\ .$$
We remark the following facts appearing in the above deformed version of the matrix (\ref{mat-Sn}):
\begin{itemize}
\item While the factors $2$ in (\ref{mat-Sn}) are deformed in $[2]$, the factors $4$ become $[2]^2$ (and not $[4]$).
\item The factors $(u+n)$ are deformed into $\al{n}$, while the factor $u^2-u+2$ is deformed into $P(\alpha)$.
\item In addition to the ``deformation rules'' described in the preceding items, powers of $q$ and of $\alpha$, which do not affect the classical limit, appear in some coefficients of the matrix. \hfill$\triangle$
\end{itemize}
\end{example}

\begin{remark}
As in Subsection \ref{subsec-rem}, other evaluations of the fusion function (\ref{phi-Hn}) leading to non-invertible elements of the Hecke algebra $H_n(q)$ are of interest in the framework of the fusion procedure for the solution (\ref{Ra-Hn}). We only indicate here the analogues of the ``non-standard'' idempotents found in Proposition \ref{prop-rem} (see (\ref{ns-idem})):
$$E^{(q)}_2:=\frac{1}{q^2+1+q^{-2}}\Psi(1,q^4,q^2)T_{w_3}^{-1}\ \ \ \text{and}\ \ \ E^{(q)}_3:=\frac{1}{q^2+1+q^{-2}}\Psi(1,q^{-4},q^{-2})T_{w_3}^{-1}\ ,$$
which are pairwise orthogonal primitive idempotents of $H_3(q)$ generating minimal left ideals isomorphic to the irreducible module $M_{(2,1)}^{H_3(q)}$.
\end{remark}


\begin{thebibliography}{99}

\bibitem{CP} V. Chari and A. Pressley, \emph{A guide to quantum groups}, Cambridge University Press (1995).

\bibitem{GRS} C. Gomez, M. Ruiz-Altaba and G. Sierra, \emph{Quantum groups in two-dimensional physics}, Cambridge University Press (1996).

\bibitem{Is} A. Isaev, \emph{Quantum groups and Yang-Baxter equation}, Preprint MPIM 04-132 (2004).

\bibitem{Ji} M. Jimbo (ed.), \emph{Yang Baxter Equation in Integrable Systems}, Adv. Series in Math. Phys. 10, World Scientific, Singapore, (1990).

\bibitem{KRS} P. Kulish, N. Reshetikhin and E. Sklyanin, \emph{Yang--Baxter equation and representation theory I}, Lett. Math. Phys. 5 (1981), 393--403.

\bibitem{BaRe} A. Babichenko and V. Regelskis, \emph{On boundary fusion and functional relations in the Baxterized affine Hecke algebra}, J. Math. Phys. 55, 043503 (2014). ArXiv:1305.1941 

\bibitem{Ch1} I. Cherednik, \emph{Some finite dimensional representations of generalized Sklyanin algebras}, Funct. Anal. Appl. 19 (1985), 77--79.

\bibitem{Ch3} I. Cherednik, \emph{On ``quantum'' deformations of irreducible finite-dimensional representations of $\mathfrak{gl}_N$},
Sov. Math. Dokl. 33 (1986), 507--510.

\bibitem{DJKMO} E. Date, M. Jimbo, A. Kuniba, T. Miwa and M. Okado, \emph{Exactly solvable SOS models II: proof of the
star-triangle relation and combinatorial identities}, Adv. Stud. Pure Math. 16 (1988), 17--122.

\bibitem{DJMO}  E. Date, M. Jimbo, T. Miwa and M. Okado, \emph{Fusion of the Eight Vertex SOS Model}, Lett. Math. Phys. 12 (1986), 209--215.

\bibitem{HZ} B-Y. Hou and Y-K. Zhou, \emph{On the fusion of face and vertex models}, J. Phys. A 22 (1989), 5089--5096

\bibitem{Ji2} M. Jimbo, \emph{Introduction to the Yang--Baxter equation}, Int. J. Mod. Phys. A 04 (1989), 3759--3777.

\bibitem{Ma} N.  MacKay, \emph{The fusion of R-matrices using the Birman--Wenzl--Murakami algebra}, J. Math. Phys. 33 (1992), 1529--1538.

\bibitem{PZ} P. Pearce and Y-K. Zhou, \emph{Fusion of ADE lattice models}, Int. J. Mod. Phys. B 8 (1994), 3531--3577. ArXiv:hep-th/9405019

\bibitem{Yu} R-H. Yue, \emph{Integrable high-spin chain related to the elliptic solution of the Yang-Baxter equation}, J. Phys. A 27 (1994), 1633--1644

\bibitem{Za} A. Zabrodin, \emph{Discrete Hirota's Equation in Quantum Integrable Models}, Int. J. Mod. Phys. B 11 (1997), 3125--3158. ArXiv:hep-th/9610039

\bibitem{Mo} A. Molev, \emph{On the fusion procedure for the symmetric group}, Reports Math. Phys. 61 (2008), 181--188. ArXiv:math/0612207

\bibitem{GW} R. Goodman and R. Wallach, \emph{Representations and invariants of the classical groups}, Cambridge University Press (1998).

\bibitem{We} H. Weyl, \emph{The classical groups, their invariants and representations}, Princeton University Press (1946).

\bibitem{Ju} A. Jucys, \emph{On the Young operators of the symmetric group}, Lietuvos Fizikos Rinkinys 6 (1966), 163--180.

\bibitem{Ch2} I. Cherednik, \emph{On special bases of irreducible finite-dimensional representations of the degenerate affine Hecke algebra}, Funct. Anal. Appl. 20 (1986), 87--89.

\bibitem{Na} M. Nazarov, \emph{Yangians and Capelli identities}, in: ``Kirillov's Seminar on Representation Theory'' (G. I. Olshanski, Ed.) Amer. Math. Soc. Transl. 181, Amer. Math. Soc., Providence, RI, (1998), 139--163. ArXiv:q-alg/9601027

\bibitem{Na2} M. Nazarov, \emph{Mixed hook-length formula for degenerate affine Hecke algebras}, Lecture Notes in Math. 1815 (2003), 223--236.
ArXiv:math/9906148 

\bibitem{Na3} M. Nazarov, \emph{A mixed hook-length formula for affine Hecke algebras}, European J. Combinatorics 25 (2004), 1345--1376. 	ArXiv:math/0307091 

\bibitem{BeRe} A. Berele and A. Regev, \emph{Hook Young diagrams with applications to combinatorics and to representations of Lie superalgebras}, Adv. Math. 64 (1987), 118--175.

\bibitem{Se} A. Sergeev, \emph{Representations of the Lie superalgebras $\mathfrak{gl}(n, m)$ and $Q(n)$ on the space of tensors},
Funct. Anal. App. 18 (1984), 80--81.

\bibitem{IMOs} A. Isaev, A. Molev and A. Os'kin, \emph{On the idempotents of Hecke algebras}, Lett. Math. Phys. 85 (2008), 79--90.  ArXiv:0804.4214

\bibitem{Ji1} M. Jimbo, \emph{A q-Analogue of $U(\mathfrak{gl}(N+1))$, Hecke algebra, and the Yang--Baxter equation}, Lett. Math. Phys. 11 (1986), 247--252.

\bibitem{Mi} H. Mitsuhashi, \emph{ Schur--Weyl reciprocity between the quantum superalgebra and the Iwahori--Hecke algebra},
Algebr. Represent. Theor. 9 (2006), 309--322. ArXiv:math/0506156

\bibitem{Moon} D. Moon,  \emph{Highest weight vectors of irreducible representations of the quantum superalgebra $U_q(gl(m,n))$}, J. Korean Math. Soc. 40 (2003), no. 1, 1--28. ArXiv:math/0105204

\bibitem{IM} A. Isaev and A. Molev, \emph{Fusion procedure for the Brauer algebra}, Algebra i Analiz, 22:3 (2010), 142--154. ArXiv:0812.4113 

\bibitem{IMO1} A. Isaev, A. Molev and O. Ogievetsky, \emph{A new fusion procedure for the Brauer algebra and evaluation
 homomorphisms}, Int. Math. Res. Not.  (2012), no. 11, 2571--2606. ArXiv:1101.1336

\bibitem{IMO2} A. Isaev, A. Molev and O. Ogievetsky, \emph{Idempotents for Birman-Murakami-Wenzl algebras and reflection equation}, Adv. Theor. Math. Phys. 18 (2011), no. 1, 1--25. ArXiv:1111.2502

\bibitem{OPdA1} O. Ogievetsky O. and L. Poulain d'Andecy, \emph{Fusion procedure for Coxeter groups of type B and complex reflection groups $G(m,1,n)$}, Proc. Amer. Math. Soc. 142 (2014), 2929--2941. ArXiv:1111.6293

\bibitem{OPdA2} O. Ogievetsky and L. Poulain d'Andecy, \emph{Fusion procedure for cyclotomic Hecke algebras},  SIGMA 10 (2014), 039, 13 pages. ArXiv:1301.4237

\bibitem{PdA} L. Poulain d'Andecy, \emph{Fusion procedure for wreath products of finite groups by the symmetric group}, Algebr. Represent. Theor. 17 (2014), no. 13, 809--830. ArXiv:1301.4399v2

\bibitem{NT} M. Nazarov and V. Tarasov, \emph{On irreducibility of tensor products of Yangian modules associated with skew Young diagrams}, Duke Math. J. 112 (2002), 342--378. ArXiv:math/0012039

\bibitem{Na4} M. Nazarov, \emph{Rational representations of Yangians associated
with skew Young diagrams}, Math. Z. 247 (2004), 21--63. ArXiv:math/0303014

\bibitem{TK} H. El Turkey and J. Kujawa, \emph{Presenting Schur superalgebra}, Pacific J. Math. 262 (2013), no. 2, 285--316. ArXiv:1209.6327

\end{thebibliography}
\end{document}